\documentclass[twocolumn,letterpaper, dvipsnames,allowtoday, accepted = 2025-12-31]{quantumarticle}

\pdfoutput=1
\usepackage[utf8]{inputenc}
\usepackage[english]{babel}
\usepackage[T1]{fontenc}
\usepackage{amsfonts}
\usepackage{amsmath}
\usepackage{amssymb}
\usepackage{mathtools}
\usepackage{graphicx}
\usepackage{caption}
\usepackage{xfrac}
\usepackage{subcaption}
\usepackage{float}
\usepackage{physics}
\usepackage{bbold}
\usepackage{tabularx}
\usepackage{multirow}
\usepackage{pbox}
\usepackage{dsfont}
\usepackage{xcolor}
\usepackage{hyperref}

\providecommand{\U}[1]{\protect\rule{.1in}{.1in}}
\newcommand{\cmmnt}[1]{}
\newtheorem{theorem}{Theorem}

\newtheorem{corollary}{Corollary}

\newtheorem{definition}{Definition}

\newtheorem{lemma}{Lemma}

\newtheorem{proposition}{Proposition}
\newtheorem{remark}{Remark}

\newenvironment{proof}[1][Proof]{\noindent\textbf{#1.} }{\ \rule{0.5em}{0.5em}}

\allowdisplaybreaks

\begin{document}
\title[ ]{No-go theorem for heralded exact one-way  key distillation}

\author{Vishal Singh}
\affiliation{School of Applied and Engineering Physics, Cornell University, Ithaca, New York 14850, USA}
\author{Mark M.~Wilde}
\affiliation{School of Electrical and Computer Engineering, Cornell University, Ithaca, New York 14850, USA}


\begin{abstract}
    The heralded exact one-way distillable secret key is equal to the largest expected rate at which perfect secret key bits can be probabilistically distilled from a bipartite state by means of local operations and one-way classical communication. Here we define the set of super two-extendible states and prove that an arbitrary  state in this set  cannot be used for heralded exact one-way secret-key distillation. This broad class of states includes both erased states and all full-rank states. Comparing the heralded exact one-way distillable secret key with the more commonly studied approximate one-way distillable secret key, our results demonstrate an extreme gap between them for many states of interest, with the approximate one-way distillable secret key being much larger. Our findings naturally extend to heralded exact one-way entanglement distillation, with similar conclusions. 
\end{abstract}
\maketitle

\section{Introduction}

Quantum key distribution has emerged as one of the most promising applications of a quantum network, as it facilitates unconditionally secure communication between distant parties~\cite{RevModPhys.92.025002,RevModPhys.94.025008,zapatero2023advances}. It allows the transmission of private data between multiple parties, such that the security is ensured by the laws of quantum mechanics~\cite{BB84, Bennett2014, Ekert91}, instead of relying on  computational assumptions about the eavesdropper~\cite{katz2007introduction}. The rapid development of quantum network technologies demands a strong understanding of our ability to distribute a secret key over a quantum network equipped with some available resources. 

Entanglement is a major ingredient that ensures quantum mechanically secure communication~\cite{Ekert91}, and it is in fact necessary as well, in the sense argued in~\cite{CLL04}. However, the ability of a quantum state to establish a secret key is not a trivial consequence of its entanglement content. Indeed, the seminal work of~\cite{HHHO05, HHHO09} established the existence of bound entangled states~\cite{HHH98} that furnish a secret key upon local measurements. This motivates a separate discussion of the privacy content in a quantum state that is clearly distinct from its entanglement content.

The distillation of secret key from a quantum state under a restricted set of operations has garnered interest~\cite{DW2005, HHHO05, Christandl06, CEHHOR07, HHHLO08, HHHO09, CSW12}, due to its practical and foundational significance in quantum information science and, particularly, quantum privacy~\cite{WTB17, QSW18}. Despite the differences, the theory of entanglement is intimately linked with quantum privacy, and a deeper understanding of one can reveal insights into the other. A specific task of interest is the distillation of secret key under local operations and one-way classical communication, abbreviated as one-way LOCC, due to its physically motivated setting and relation with the private capacity of quantum channels~\cite{CWY04, Dev05}.

In this paper, we analyze the heralded exact one-way secret-key distillation  approach
in which a perfect secret key is distilled from an initial bipartite state, albeit probabilistically. We also refer to this setting as probabilistic secret-key distillation with zero error, and we adopt the latter terminology henceforth. In particular, we study the one-way distillable secret key of a bipartite state in the
probabilistic setting, which is roughly defined as the maximum achievable expected rate of distilling secret key bits from an arbitrarily large number
of copies of the state using one-way LOCC channels. 

Sometimes the term \emph{probabilistic resource distillation} is used for probabilistic approximate distillation of a resource in the asymptotic setting, for example in~\cite{Regula22,Regula22b}. In that setting, a sequence of probabilistic protocols is considered a valid probabilistic resource distillation protocol if the sequence of distilled states, when the protocol is successful, converges to the maximal-resource state asymptotically. We emphasize that we use the term \emph{probabilistic key distillation} in a more strict sense, where the distillation protocol must yield a perfect secret key whenever the protocol is successful, even if it requires an arbitrarily large number of resource states.

The particular contributions of our paper are as follows. We first establish the definition of probabilistic one-way distillable secret key, which is fundamentally different from approximate one-way distillable secret key~\cite{DW2005}. We find a set of states,  called the set of \textit{super two-extendible} states, which have no probabilistic one-way distillable secret key and can also be described via semidefinite constraints. Using the examples of erased states and full-rank states, we show that there exists an extreme gap between the probabilistic one-way distillable secret key and the approximate one-way distillable secret key for several states, with the former  being equal to zero while the latter is strictly non-zero for these examples. 

Our results establish fundamental limitations on probabilistic secret-key distillation, and consequently on probabilistic entanglement distillation, under one-way LOCC channels. The class of super two-extendible states provides a computationally-friendly framework for analyzing secret-key distillation in a resource-theoretic setting, due to its semidefinite characterization. Furthermore, our results emphasize the importance of allowing some error in secret-key distillation, as doing so can facilitate key distillation from otherwise undistillable states. 

One of the main tools that we use to establish our results is the resource theory of unextendibility~\cite{KDWW19, KDWW21} and its state-dependent variation~\cite{WWW21}. This resource theory was developed in~\cite{KDWW19, KDWW21} as a relaxation of the resource theory of entanglement, in which one-way LOCC channels are allowed for free. An important quantity that we employ is the min-unextendible entanglement of a bipartite state~\cite{WWW21}. Our work also provides significant improvements on existing bounds~\cite[Sec.~V.D]{WWW21} regarding the overhead of probabilistic one-way secret key distillation. 

\section{One-way secret-key distillation}

Let us begin by considering the task of secret-key distillation. The objective of such a  protocol $\mathcal{L}_{AB\to A'B'}$ is to transform a bipartite state~$\rho_{AB}$, purified by $\psi_{ABE}$, into a tripartite key state as follows:
\begin{equation}
    \mathcal{L}_{AB\to A'B'}\!\left(\psi_{ABE}\right) = \frac{1}{k} \sum_{i=0}^{k-1} |i\rangle\!\langle i|_{A'}\otimes |i\rangle\!\langle i|_{B'} \otimes \sigma_E,
\end{equation}
where $\sigma_E$ is an arbitrary quantum state. Alice and Bob, holding systems $A'$ and $B'$ respectively, can use the classically correlated state shared between them to communicate a message of $\log_2 k$ bits using the one-time-pad scheme. Any eavesdropper holding the system $E$ cannot decipher anything about the message because $\sigma_E$ is independent of the symbol $i$, hence, ensuring secrecy of the communication. 

In~\cite{HHHO05, HHHO09}, it has been shown that the distillation of a secret key is equivalent to the distillation of a bipartite private state of the following form:
\begin{equation}\label{eq:bip_private_state_defn}
    \gamma^k_{A_0A_1B_0B_1} \coloneqq \frac{1}{k}\sum_{i,j=0}^{k-1} |ii\rangle\!\langle jj|_{A_0B_0}\otimes U_i\omega_{A_1B_1}U_j^{\dagger},
\end{equation}
where $\omega_{A_1B_1}$ is an arbitrary quantum state and $\left(U_i\right)_{i=0}^{k-1}$ is a tuple of unitary operators. The systems $A_0B_0$ are the key systems, and $A_1B_1$ are the shield systems. The bipartite private state defined in \eqref{eq:bip_private_state_defn} can be used to distill at least $\log_2 k$ secret key bits. We can hence reframe the task of secret key distillation into the distillation of bipartite private states. From here on we simplify the labeling for the key systems and shield systems: when referring to a private state $\gamma^k_{A_0A_1B_0B_1}$, we group the systems held by Alice into a single system label $A\coloneqq A_0A_1$ and all the systems held by Bob into a single system label $B\coloneqq B_0B_1$.

In a probabilistic one-way secret-key distillation protocol, Alice and Bob use local operations and one-way classical communication from Alice to Bob to distill a secret key, or equivalently a bipartite private state, from a shared resource state $\rho_{AB}$ with some probability $p\in [0,1]$. The distillation process can be mathematically described as the action of a one-way LOCC channel $\mathcal{L}^{\to}_{AB\to XA'B'}$ on the resource state $\rho_{AB}$ as follows:
\begin{equation}\label{eq:private_st_distillation}
    \mathcal{L}^{\to}_{AB\to XA'B'}\!\left(\rho_{AB}\right) = p[1]_X\otimes \gamma^k_{A'B'} + (1-p)[0]_X\otimes \sigma_{A'B'},
\end{equation}
where we have used the shorthand
\begin{equation}
[i]\coloneqq |i\rangle\!\langle i|.    
\end{equation}
In \eqref{eq:private_st_distillation}, $\gamma^k_{A'B'}$ is a bipartite private state with at least $\log_2 k$ bits of secrecy, system $X$ is a classical flag indicating the success or failure of the protocol, and $\sigma_{A'B'}$ is an arbitrary quantum state generated when the protocol fails to distill a private state.  

In an arbitrary probabilistic secret-key distillation protocol, both parties must have access to the flag~$X$ in order to use the distilled key for private communication. Suppose that a probabilistic secret-key distillation protocol fails to establish a secret key. In that case, both parties involved in the distillation process can discard their systems and repeat the protocol with another instance of the resource state. If one-way LOCC channels are available for free, it suffices to demand that Alice receives the flag $X$ since she can send the flag to Bob using the freely available forward classical channel.

We can now quantify the resource in a bipartite state that is relevant for the task of secret-key distillation using one-way LOCC channels. For this purpose, we define the probabilistic one-way distillable secret key of a bipartite state as follows:
\begin{definition}
    The probabilistic one-way distillable secret key of a bipartite state $\rho_{AB}$ is the maximum expected rate at which secret key bits can be distilled from a bipartite state using one-way LOCC channels. It is formally defined as
    \begin{equation}
        K_{D}^{\to}(\rho_{AB}) \coloneqq \liminf_{n\to \infty} \frac{1}{n} K_{D}^{(1),\to}(\rho_{AB}^{\otimes n}),
    \end{equation}
    where the one-shot probabilistic one-way distillable secret key is defined as
    \begin{multline}\label{eq:prob_1W_distill_key_defn}
        K_{D}^{(1),\to}(\rho_{AB}) \coloneqq\\
        \sup_{\substack{p \in [0,1], k\in \mathbb{N}\\ \mathcal{L}^{\to} \in \operatorname{1WL}, \gamma^k_{A'B'}}}\left\{
        \begin{array}{c}
             p \log_2 k:\\
             \mathcal{L}^{\to}\!\left(\rho_{AB}\right) =
             p[1]_{X_A}\otimes\gamma^k_{A'B'} \\
             + (1-p)[0]_{X_A}\otimes\frac{I_{A'B'}}{d_A'd_B'}
        \end{array}
        \right\}.
    \end{multline}
    In the above, $\mathcal{L}^{\to}_{AB\to X_A A'B'}$ is a one-way LOCC channel, 1WL denotes the set of all one-way LOCC channels, the optimization is over every private state $ \gamma^k_{A'B'}$ of $\log_2 k$ secret key bits, $X_A$ is a classical flag held by Alice, and $I_{A'B'}$ is the identity operator.
\end{definition}

In \eqref{eq:prob_1W_distill_key_defn} we require that the state generated upon failure of the protocol is a maximally mixed state. This additional constraint on secret-key distillation protocols does not affect the maximum expected number of secret key bits that can be distilled from a bipartite state using one-way LOCC channels, as we show in Appendix~\ref{app:prob_dist_key_eq_defn}.

It is worthwhile to note that a maximally entangled state of Schmidt rank $k$ is a  private state holding $\log_2 k$ secret key bits~\cite{HHHO05, HHHO09}. As such, an arbitrary entanglement distillation protocol can be transformed into a secret-key distillation protocol without affecting the rate of distillation. Hence, the one-way distillable secret key of a quantum state is not less than the one-way distillable entanglement of the state in both the probabilistic and approximate settings. 

A simpler way to analyze the probabilistic distillation of secret keys under the action of one-way LOCC channels is by considering the erasure symbol. The erasure symbol $[e]$ is defined to be orthogonal to every state in the span of $\{|i\rangle\!\langle j|\}_{i,j=0}^{d-1}$, where $d$ is the dimension of the underlying system. For practical purposes, one can think of the erasure symbol as a pure state in the $(d+1)$-dimensional Hilbert space that is orthogonal to every state in the $d$-dimensional Hilbert space representing the system of interest. In the case of a joint system $S$ that comprises multiple subsystems, say, $S_1,S_2,\ldots,S_k$, the erasure symbol~$[e]_S$ can be represented as follows:
\begin{equation}
    [e]_S = [e]_{S_1}\otimes[e]_{S_2}\otimes\cdots\otimes[e]_{S_k},
\end{equation}
which is orthogonal to every state on the system $S$, and hence, it is consistent with the definition of the erasure symbol. 

In a one-way secret-key distillation protocol, if Alice finds the flag $X_A$ in the state $[0]_{X_A}$, she can erase her state, and she can instruct Bob to erase his state as well. We call the resulting state the doubly erased private state, which has the following form:  
\begin{equation}\label{eq:double_eras_priv_st}
    \eta^{p,k}_{AB} \coloneqq p~\gamma^k_{AB} + (1-p)[e]_A\otimes [e]_B,
\end{equation}
where both Alice and Bob can retrieve the flag by performing the POVM~$\{\Pi,[e]\}$ on their respective systems, where $\Pi \coloneqq \sum_{i=0}^{d-1} [i]$. They can further replace their state with a maximally mixed state upon measuring the erasure symbol, hence retrieving the distilled state in~\eqref{eq:prob_1W_distill_key_defn}. That is,
\begin{equation}\label{eq:2_eras_st_distill_st_tf}
    \eta^{p,k}_{AB} \xleftrightarrow[]{\operatorname{LO}} p[1]_{X_A}\otimes\gamma^k_{A'B'} + (1-p)[0]_{X_A}\otimes \frac{I_{A'B'}}{d_{A'}d_{B'}}.
\end{equation}
Since the transformation in \eqref{eq:2_eras_st_distill_st_tf} can be effected by one-way LOCC only, the distillation of the doubly erased private state~$\eta^{p,k}_{AB}$ is equivalent to the distillation of $\log_2 k$ secret key bits with probability $p$.

\section{Min-unextendible entanglement}

The min-unextendible entanglement of a bipartite state was defined in~\cite[Sec.~4.2]{WWW21} as a monotone for the state-dependent resource theory of unextendibility. We briefly mention the relevant properties of this quantity here, and we refer the reader to~\cite{WWW21} for a detailed presentation of the resource theory of unextendibility.

The min-unextendible entanglement is defined with respect to the min-relative entropy~\cite[Def.~2]{Dat09} as follows:
\begin{equation}
    E^u_{\operatorname{min}}\!\left(\rho_{AB}\right) \coloneqq \inf_{\sigma_{AB}\in \mathcal{F}\left(\rho_{AB}\right)}-\frac{1}{2} \log_2 \operatorname{Tr}\!\left[\Pi^{\rho}_{AB}\sigma_{AB}\right],
\end{equation}
where $\Pi^{\rho}_{AB}$ is the projection onto the support of $\rho_{AB}$ and the optimization is over all states in the following set:
\begin{equation}\label{eq:extension_set}
     \mathcal{F}\!\left(\rho_{AB}\right) \coloneqq \left\{
     \begin{array}{c}
         \operatorname{Tr}_{B}\!\left[\omega_{ABB'}\right]:  
         
          \rho_{AB} = \operatorname{Tr}_{B'}\!\left[\omega_{ABB'}\right],\\
          \omega_{ABB'} \in \mathcal{S}\!\left(ABB'\right)
     \end{array}
         \right\},
 \end{equation}
with $\mathcal{S}\!\left(ABB'\right)$ the set of all states of the joint system $ABB'$ and system $B'$ being isomorphic to system $B$. 

The min-unextendible entanglement of a bipartite state has several properties relevant to our discussion. Firstly, it is non-negative, and it is additive with respect to tensor products of states. It is monotonic under the action of two-extendible channels, as defined in~\cite{KDWW19, KDWW21}, which is a superset of one-way LOCC channels. As such, the min-unextendible entanglement of a bipartite state does not increase under the action of one-way LOCC channels. Lastly, the min-unextendible entanglement of a bipartite private state is not less than the number of secret key bits held by the private state. See Appendix~\ref{app:min_unext_ent} for more details. 

\begin{remark}\label{rem:min_uenxt_ent_eq_tf}
    The min-unextendible entanglement of a bipartite state does not increase under one-way LOCC channels. Hence, the transformation in \eqref{eq:2_eras_st_distill_st_tf} implies the following equality:
    \begin{multline}
        E^u_{\min}\!\left(\eta^{p,k}_{AB}\right) =\\ E^u_{\min}\!\left(p[1]_{X_A}\otimes \gamma^k_{AB} + (1-p)[0]_{X_A}\otimes\frac{I_{AB}}{d_Ad_B}\right).
    \end{multline}
\end{remark}

\section{Limitations on probabilistic one-way distillable secret key}

The monotonicity of the min-unextendible entanglement of a bipartite state under the action of one-way LOCC channels implies that the min-unextendible entanglement of the target state in a probabilistic one-way secret-key distillation protocol does not exceed the min-unextendible entanglement of the source state. We first present a lower bound on the min-unextendible entanglement of the doubly erased private state, which is the target state of a probabilistic one-way secret-key distillation protocol.

\begin{lemma}\label{lem:min_unext_ent_2_eras_eq}
    For all $p \in [0,1]$ and every integer $k \geq 2$, the min-unextendible entanglement of a doubly erased private state $\eta^{p,k}_{AB}$ is bounded from below by the following quantity:
    \begin{equation}
        E^u_{\min}\!\left(\eta^{p,k}_{AB}\right) \ge -\frac{1}{2}\log_2\!\left(\frac{p}{k^2} + 1-p\right)  .
    \end{equation}
\end{lemma}
\begin{proof}
    See Appendix~\ref{app:min_unext_ent_2_eras_eq}.
\end{proof}

\medskip

The min-unextendible entanglement of the doubly erased private state $\eta^{p,k}_{AB}$  is strictly positive for all $p \in (0,1]$ and every integer $k\ge 2$, and it is only equal to zero at $p=0$ because the resulting state is a product state (\cite[Proposition~3]{WWW21}). As a consequence of the one-way LOCC monotonicity of the min-unextendible entanglement, a state whose min-unextendible entanglement is equal to zero cannot be used to distill any number of secret key bits with a non-zero probability using one-way LOCC channels. In what follows, we identify the structure of the set of states for which the min-unextendible entanglement is equal to zero and state the main result of this work formally in Theorem~\ref{theo:prob_distill_ads_eq_0}.

Let us first analyze the set of states whose min-unextendible entanglement is equal to zero. If there exists a state $\sigma_{AB} \in \mathcal{F}\!\left(\rho_{AB}\right)$ such that $\rho_{AB} = \sigma_{AB}$,
then $\rho_{AB}$ is a two-extendible state~\cite{Wer89, DPS02, DPS04}, and its min-unextendible entanglement is equal to zero. A more general set of states for which the min-unextendible entanglement is equal to zero, which we call super two-extendible states, can be defined as follows: 
\begin{equation}
    \operatorname{2-EXT}_{\operatorname{sup}}\!\left(A\!:\!B\right) \coloneqq  \left\{\begin{array}{c}
         \rho_{AB}:  
         
         \exists \sigma_{AB}\in \mathcal{F}\!\left(\rho_{AB}\right),\\
         \operatorname{supp}\!\left(\sigma_{AB}\right) \subseteq \operatorname{supp}\!\left(\rho_{AB}\right)\\
         
    \end{array}\right\}.
\end{equation}
The set of super two-extendible states is convex but not closed (see Appendix~\ref{app:sup_2_ext_convex_open}).

\begin{proposition}\label{prop:ads_min_unext_ent_eq_0}
    The min-unextendible entanglement of a quantum state is equal to zero if and only if it is super two-extendible. 
\end{proposition}
\begin{proof}
    See Appendix~\ref{app:ads_min_unext_ent_eq_0}.
\end{proof}

\medskip

While the approximate one-way distillable secret key of a two-extendible state, also known as an anti-degradable state~\cite{LDS18}, is equal to zero~\cite[Thm.~15.43]{KW24}, the same is not true for a general super two-extendible state, as we shall see later in this paper. However, 
a super two-extendible state cannot be used for probabilistic one-way secret key distillation since its min-unextendible entanglement is equal to zero. Combining this fact with the additive property of the min-unextendible entanglement, we arrive at our main no-go theorem stated as Theorem~\ref{theo:prob_distill_ads_eq_0}, which identifies a broad set of states for which the probabilistic one-way distillable secret key is equal to zero.

\begin{theorem}\label{theo:prob_distill_ads_eq_0}
    The probabilistic one-way distillable secret key of a super two-extendible state is equal to zero.
\end{theorem}

\begin{proof}
    See Appendix~\ref{app:prob_distill_eras_st_eq_0_proof}.
\end{proof}

\subsection{Erased private state}

Let us consider the following state, which is established if Bob's share of a bipartite private state~$\gamma^k_{AB}$ gets erased with probability $1-p$:
\begin{equation}\label{eq:eras_priv_st_defn}
    \widetilde{\eta}^{p,k}_{AB} \coloneqq p~\gamma^k_{AB} + (1-p)\operatorname{Tr}_B\!\left[\gamma^k_{AB}\right]\otimes [e]_B,
\end{equation}
where $\operatorname{Tr}_B\!\left[\cdot\right]$ refers to the partial trace of the argument over system~$B$. We call the state in \eqref{eq:eras_priv_st_defn} an erased private state.

\begin{proposition}\label{lem:eras_min_unext_ent_eq_0}
    For all $p \in [0,1)$ and every integer $k \geq 2$, the erased private state $\widetilde{\eta}^{p,k}_{AB}$ is a super two-extendible state and thus has probabilistic one-way distillable secret key equal to zero. 
\end{proposition}

\begin{proof}
    See Appendix~\ref{app:eras_min_unext_ent_eq_0}.
\end{proof}

\medskip

An erased state is defined as follows:
\begin{equation}\label{eq:eras_st_defn}
    \widetilde{\Phi}^{p,d}_{AB} \coloneqq p~\Phi^d_{AB} + (1-p)\frac{I_A}{d_A}\otimes [e]_B,
\end{equation}
where $\Phi^d_{AB}$ is the maximally entangled state with Schmidt rank equal to $d$. Since the erased state is a special case of an erased private state, it is in the set of super two-extendible states for all $p\in [0,1)$, which leads to Corollary~\ref{cor:prob_distill_eras_st_eq_0} stated below. 

\begin{corollary}
\label{cor:prob_distill_eras_st_eq_0}
    For all $p \in [0,1)$ and every integer $d \geq 2$, the probabilistic one-way distillable secret key of the erased state~$\widetilde{\Phi}^{p,d}_{AB}$ is equal to zero.
\end{corollary}

The set of full-rank states also lies inside the set of super two-extendible states, which leads to  Corollary~\ref{cor:dist_full_rank} below.

\begin{corollary}\label{cor:dist_full_rank}
    All full-rank states are super two-extendible states, and the probabilistic one-way distillable secret key of such quantum states is equal to zero.
\end{corollary}

\begin{proof}
    See Appendix~\ref{app:dist_full_rank}.
\end{proof}

\medskip

The inability to probabilistically distill resources from full-rank states has been observed in general quantum resource theories~\cite{Kent98, FL20, Regula22b}. Corollary~\ref{cor:dist_full_rank} extends this result to probabilistic secret-key distillation under one-way LOCC channels. This result also implies that the overhead of probabilistic secret-key distillation considered in~\cite[Sec.~V-B]{WWW21} is infinite for quantum states with full-rank density matrices, thus providing a significant strengthening of the bounds from~\cite[Sec.~V-B]{WWW21}.

\subsection{Approximate versus probabilistic secret key distillation}

Given the inability to distill a secret key from super two-extendible states even probabilistically, it is natural to ask if one can distill a key from such states if an arbitrarily small error is allowed. The rate at which secret-key bits can be distilled from a state such that the error vanishes asymptotically, using a one-way LOCC protocol, is called the one-way distillable key of the state~\cite{DW2005}, or the approximate one-way distillable key of the state, as we call it in this paper to clearly distinguish it from the probabilistic one-way distillable key of the state.

Indeed, it is possible distill a secret key from a state when an asymptotically vanishing error is allowed in distillation even if a perfect key cannot be distilled with any non-zero probability from the states, as we can see from the example of erased states. The approximate one-way distillable secret key is not less than the approximate one-way distillable entanglement of a state, which in turn is bounded from below by the coherent information of the state~\cite[Thm.~10]{DW2005}. For $p > \frac{1}{2} $, the coherent information of an erased state $\widetilde{\Phi}^{p,d}_{AB}$  is strictly non-negative. Similarly, the coherent information of some isotropic~\cite{HH99} and Werner states~\cite{Wer89_HV}, which are full-rank states, is non-zero (see Appendices~\ref{app:coherent_info_Werner} and~\ref{app:coherent_info_iso_proof}). However, Corollaries~\ref{cor:prob_distill_eras_st_eq_0} and~\ref{cor:dist_full_rank} state that the probabilistic one-way distillable secret key of these states is equal to zero, demonstrating an extreme gap between the probabilistic and approximate one-way distillable secret key.

\section{Conclusion}

We found a convex set of states, that we call super two-extendible states, from which it is impossible to probabilistically distill perfect secret-key bits with any non-zero probability. Our main result is a no-go theorem stating that the probabilistic one-way distillable secret key of a super two-extendible state, and consequently its probabilistic one-way distillable entanglement, is equal to zero.

We demonstrated an extreme gap between the probabilistic and approximate one-way distillable secret key for some  states. As such, we showed that probabilistic distillation of perfect secret keys using one-way LOCC protocols is not feasible for most states of interest, and one must allow some error for practically reasonable key distillation protocols. Considering the intermediate regime between probabilistic and approximate distillation~\cite{FL20, Regula22, EW22}, where the goal is to distill high-fidelity maximal-resource states with a non-zero probability, would shed light on the trade-off between distillation error and success probability in the task of probabilistic distillation of noisy secret keys.

\section*{Acknowledgments}
We are grateful to Kaiyuan Ji, Ludovico Lami, Felix Leditzky, Theshani Nuradha, Dhrumil Patel, Aby Philip, Bartosz Regula, and Soorya Rethinasamy for insightful discussions. We are especially grateful to  Ludovico Lami and Bartosz Regula for feedback that greatly improved our paper. We acknowledge financial support from NSF Grant No.~2315398.

\bibliographystyle{quantum}
\bibliography{Ref}

\begin{thebibliography}{10}

\bibitem{RevModPhys.92.025002}
Feihu Xu, Xiongfeng Ma, Qiang Zhang, Hoi-Kwong Lo, and Jian-Wei Pan.
\newblock ``Secure quantum key distribution with realistic devices''.
\newblock \href{https://dx.doi.org/10.1103/RevModPhys.92.025002}{Reviews of Modern Physics {\bf 92}, 025002}~(2020).

\bibitem{RevModPhys.94.025008}
Christopher Portmann and Renato Renner.
\newblock ``Security in quantum cryptography''.
\newblock \href{https://dx.doi.org/10.1103/RevModPhys.94.025008}{Reviews of Modern Physics {\bf 94}, 025008}~(2022).

\bibitem{zapatero2023advances}
V{\'\i}ctor Zapatero, Tim van Leent, Rotem Arnon-Friedman, Wen-Zhao Liu, Qiang Zhang, Harald Weinfurter, and Marcos Curty.
\newblock ``Advances in device-independent quantum key distribution''.
\newblock \href{https://dx.doi.org/10.1038/s41534-023-00684-x}{npj Quantum Information {\bf 9}, 10}~(2023).

\bibitem{BB84}
Charles~H. Bennett and Gilles Brassard.
\newblock ``Quantum cryptography: Public key distribution and coin tossing''.
\newblock In Proceedings of IEEE International Conference on Computers, Systems, and Signal Processing.
\newblock Page 175.
\newblock India~(1984).

\bibitem{Bennett2014}
Charles~H. Bennett and Gilles Brassard.
\newblock ``Quantum cryptography: Public key distribution and coin tossing''.
\newblock \href{https://dx.doi.org/https://doi.org/10.1016/j.tcs.2014.05.025}{Theoretical Computer Science {\bf 560}, 7--11}~(2014).

\bibitem{Ekert91}
Artur~K. Ekert.
\newblock ``Quantum cryptography based on {B}ell's theorem''.
\newblock \href{https://dx.doi.org/10.1103/PhysRevLett.67.661}{Physical Review Letters {\bf 67}, 661--663}~(1991).

\bibitem{katz2007introduction}
Jonathan Katz and Yehuda Lindell.
\newblock ``Introduction to modern cryptography: Principles and protocols''.
\newblock \href{https://dx.doi.org/https://doi.org/10.1201/9781420010756}{Chapman and Hall/CRC}. ~(2007).

\bibitem{CLL04}
Marcos Curty, Maciej Lewenstein, and Norbert L\"utkenhaus.
\newblock ``Entanglement as a precondition for secure quantum key distribution''.
\newblock \href{https://dx.doi.org/10.1103/PhysRevLett.92.217903}{Physical Review Letters {\bf 92}, 217903}~(2004).

\bibitem{HHHO05}
Karol Horodecki, Micha\l{} Horodecki, Pawe\l{} Horodecki, and Jonathan Oppenheim.
\newblock ``Secure key from bound entanglement''.
\newblock \href{https://dx.doi.org/10.1103/PhysRevLett.94.160502}{Physical Review Letters {\bf 94}, 160502}~(2005).
\newblock  \href{http://arxiv.org/abs/quant-ph/0309110}{arXiv:quant-ph/0309110}.

\bibitem{HHHO09}
Karol Horodecki, Micha\l{} Horodecki, Pawe\l{} Horodecki, and Jonathan Oppenheim.
\newblock ``General paradigm for distilling classical key from quantum states''.
\newblock \href{https://dx.doi.org/10.1109/TIT.2008.2009798}{IEEE Transactions on Information Theory {\bf 55}, 1898--1929}~(2009).
\newblock  \href{http://arxiv.org/abs/quant-ph/0506189}{arXiv:quant-ph/0506189}.

\bibitem{HHH98}
Micha\l{} Horodecki, Pawe\l{} Horodecki, and Ryszard Horodecki.
\newblock ``Mixed-state entanglement and distillation: Is there a ``bound'' entanglement in nature?''.
\newblock \href{https://dx.doi.org/10.1103/PhysRevLett.80.5239}{Physical Review Letters {\bf 80}, 5239--5242}~(1998).
\newblock  \href{http://arxiv.org/abs/quant-ph/9801069}{arXiv:quant-ph/9801069}.

\bibitem{DW2005}
Igor Devetak and Andreas Winter.
\newblock ``Distillation of secret key and entanglement from quantum states''.
\newblock \href{https://dx.doi.org/10.1098/rspa.2004.1372}{Proceedings of the Royal Society A {\bf 461}, 207--235}~(2005).
\newblock  \href{http://arxiv.org/abs/quant-ph/0306078}{arXiv:quant-ph/0306078}.

\bibitem{Christandl06}
Matthias Christandl.
\newblock ``The structure of bipartite quantum states---insights from group theory and cryptography''~(2006).
\newblock  \href{http://arxiv.org/abs/quant-ph/0604183}{arXiv:quant-ph/0604183}.

\bibitem{CEHHOR07}
Matthias Christandl, Artur Ekert, Micha{\l} Horodecki, Pawe{\l} Horodecki, Jonathan Oppenheim, and Renato Renner.
\newblock ``Unifying classical and quantum key distillation''.
\newblock In Salil~P. Vadhan, editor, Theory of Cryptography.
\newblock Pages 456--478.
\newblock Berlin, Heidelberg~(2007). Springer Berlin Heidelberg.
\newblock  \href{http://arxiv.org/abs/quant-ph/0608199}{arXiv:quant-ph/0608199}.

\bibitem{HHHLO08}
Karol Horodecki, Micha\l{} Horodecki, Pawe\l{} Horodecki, Debbie Leung, and Jonathan Oppenheim.
\newblock ``Quantum key distribution based on private states: Unconditional security over untrusted channels with zero quantum capacity''.
\newblock \href{https://dx.doi.org/10.1109/TIT.2008.921870}{IEEE Transactions on Information Theory {\bf 54}, 2604--2620}~(2008).
\newblock  \href{http://arxiv.org/abs/quant-ph/0608195}{arXiv:quant-ph/0608195}.

\bibitem{CSW12}
Matthias Christandl, Norbert Schuch, and Andreas Winter.
\newblock ``Entanglement of the antisymmetric state''.
\newblock \href{https://dx.doi.org/10.1007/s00220-012-1446-7}{Communications in Mathematical Physics {\bf 311}, 397--422}~(2012).
\newblock  \href{http://arxiv.org/abs/0910.4151}{arXiv:0910.4151}.

\bibitem{WTB17}
Mark~M. Wilde, Marco Tomamichel, and Mario Berta.
\newblock ``Converse bounds for private communication over quantum channels''.
\newblock \href{https://dx.doi.org/10.1109/TIT.2017.2648825}{IEEE Transactions on Information Theory {\bf 63}, 1792--1817}~(2017).
\newblock  \href{http://arxiv.org/abs/1602.08898}{arXiv:1602.08898}.

\bibitem{QSW18}
Haoyu Qi, Kunal Sharma, and Mark~M Wilde.
\newblock ``Entanglement-assisted private communication over quantum broadcast channels''.
\newblock \href{https://dx.doi.org/10.1088/1751-8121/aad5f3}{Journal of Physics A: Mathematical and Theoretical {\bf 51}, 374001}~(2018).
\newblock  \href{http://arxiv.org/abs/1803.03976}{arXiv:1803.03976}.

\bibitem{CWY04}
N.~Cai, A.~Winter, and R.~W. Yeung.
\newblock ``Quantum privacy and quantum wiretap channels''.
\newblock \href{https://dx.doi.org/10.1007/s11122-005-0002-x}{Problems of Information Transmission {\bf 40}, 318--336}~(2004).

\bibitem{Dev05}
I.~Devetak.
\newblock ``The private classical capacity and quantum capacity of a quantum channel''.
\newblock \href{https://dx.doi.org/10.1109/TIT.2004.839515}{IEEE Transactions on Information Theory {\bf 51}, 44--55}~(2005).
\newblock  \href{http://arxiv.org/abs/quant-ph/0304127}{arXiv:quant-ph/0304127}.

\bibitem{Regula22}
Bartosz Regula.
\newblock ``Probabilistic transformations of quantum resources''.
\newblock \href{https://dx.doi.org/10.1103/PhysRevLett.128.110505}{Physical Review Letters {\bf 128}, 110505}~(2022).
\newblock  \href{http://arxiv.org/abs/2109.04481}{arXiv:2109.04481}.

\bibitem{Regula22b}
Bartosz Regula.
\newblock ``Tight constraints on probabilistic convertibility of quantum states''.
\newblock \href{https://dx.doi.org/10.22331/q-2022-09-22-817}{{Quantum} {\bf 6}, 817}~(2022).
\newblock  \href{http://arxiv.org/abs/2112.11321}{arXiv:2112.11321}.

\bibitem{KDWW19}
Eneet Kaur, Siddhartha Das, Mark~M. Wilde, and Andreas Winter.
\newblock ``Extendibility limits the performance of quantum processors''.
\newblock \href{https://dx.doi.org/10.1103/PhysRevLett.123.070502}{Physical Review Letters {\bf 123}, 070502}~(2019).
\newblock  \href{http://arxiv.org/abs/2108.03137}{arXiv:2108.03137}.

\bibitem{KDWW21}
Eneet Kaur, Siddhartha Das, Mark~M. Wilde, and Andreas Winter.
\newblock ``Resource theory of unextendibility and nonasymptotic quantum capacity''.
\newblock \href{https://dx.doi.org/10.1103/PhysRevA.104.022401}{Physical Review A {\bf 104}, 022401}~(2021).
\newblock  \href{http://arxiv.org/abs/1803.10710}{arXiv:1803.10710}.

\bibitem{WWW21}
Kun Wang, Xin Wang, and Mark~M Wilde.
\newblock ``Quantifying the unextendibility of entanglement''.
\newblock \href{https://dx.doi.org/10.1088/1367-2630/ad264e}{New Journal of Physics {\bf 26}, 033013}~(2024).
\newblock  \href{http://arxiv.org/abs/1911.07433}{arXiv:1911.07433}.

\bibitem{Dat09}
Nilanjana Datta.
\newblock ``Min- and max-relative entropies and a new entanglement monotone''.
\newblock \href{https://dx.doi.org/10.1109/TIT.2009.2018325}{IEEE Transactions on Information Theory {\bf 55}, 2816--2826}~(2009).
\newblock  \href{http://arxiv.org/abs/0803.2770}{arXiv:0803.2770}.

\bibitem{Wer89}
Reinhard~F. Werner.
\newblock ``An application of {B}ell's inequalities to a quantum state extension problem''.
\newblock \href{https://dx.doi.org/10.1007/BF00399761}{Letters in Mathematical Physics {\bf 17}, 359--363}~(1989).

\bibitem{DPS02}
A.~C. Doherty, Pablo~A. Parrilo, and Federico~M. Spedalieri.
\newblock ``Distinguishing separable and entangled states''.
\newblock \href{https://dx.doi.org/10.1103/PhysRevLett.88.187904}{Physical Review Letters {\bf 88}, 187904}~(2002).
\newblock  \href{http://arxiv.org/abs/quant-ph/0112007}{arXiv:quant-ph/0112007}.

\bibitem{DPS04}
Andrew~C. Doherty, Pablo~A. Parrilo, and Federico~M. Spedalieri.
\newblock ``Complete family of separability criteria''.
\newblock \href{https://dx.doi.org/10.1103/PhysRevA.69.022308}{Physical Review A {\bf 69}, 022308}~(2004).
\newblock  \href{http://arxiv.org/abs/quant-ph/0308032}{arXiv:quant-ph/0308032}.

\bibitem{LDS18}
Felix Leditzky, Nilanjana Datta, and Graeme Smith.
\newblock ``Useful states and entanglement distillation''.
\newblock \href{https://dx.doi.org/10.1109/TIT.2017.2776907}{IEEE Transactions on Information Theory {\bf 64}, 4689--4708}~(2018).
\newblock  \href{http://arxiv.org/abs/1701.03081}{arXiv:1701.03081}.

\bibitem{KW24}
Sumeet Khatri and Mark~M. Wilde.
\newblock ``Principles of quantum communication theory: A modern approach''~(2024).
\newblock  \href{http://arxiv.org/abs/2011.04672v2}{arXiv:2011.04672v2}.

\bibitem{Kent98}
Adrian Kent.
\newblock ``Entangled mixed states and local purification''.
\newblock \href{https://dx.doi.org/10.1103/PhysRevLett.81.2839}{Physical Review Letters {\bf 81}, 2839--2841}~(1998).

\bibitem{FL20}
Kun Fang and Zi-Wen Liu.
\newblock ``No-go theorems for quantum resource purification''.
\newblock \href{https://dx.doi.org/10.1103/PhysRevLett.125.060405}{Physical Review Letters {\bf 125}, 060405}~(2020).
\newblock  \href{http://arxiv.org/abs/1909.02540}{arXiv:1909.02540}.

\bibitem{HH99}
Micha\l{} Horodecki and Pawe\l{} Horodecki.
\newblock ``Reduction criterion of separability and limits for a class of distillation protocols''.
\newblock \href{https://dx.doi.org/10.1103/PhysRevA.59.4206}{Physical Review A {\bf 59}, 4206--4216}~(1999).
\newblock  \href{http://arxiv.org/abs/quant-ph/9708015}{arXiv:quant-ph/9708015}.

\bibitem{Wer89_HV}
Reinhard~F. Werner.
\newblock ``Quantum states with {Einstein-Podolsky-Rosen} correlations admitting a hidden-variable model''.
\newblock \href{https://dx.doi.org/10.1103/PhysRevA.40.4277}{Physical Review A {\bf 40}, 4277--4281}~(1989).

\bibitem{EW22}
Jens Eisert and Mark~M. Wilde.
\newblock ``A smallest computable entanglement monotone''.
\newblock In 2022 IEEE International Symposium on Information Theory (ISIT).
\newblock \href{https://dx.doi.org/10.1109/ISIT50566.2022.9834375}{Pages 2439--2444}.
\newblock ~(2022).
\newblock  \href{http://arxiv.org/abs/2201.00835}{arXiv:2201.00835}.

\bibitem{polyanskiy2010arimoto}
Yury Polyanskiy and Sergio Verd{\'u}.
\newblock ``Arimoto channel coding converse and {R\'e}nyi divergence''.
\newblock In 2010 48th Annual Allerton Conference on Communication, Control, and Computing.
\newblock \href{https://dx.doi.org/10.1109/ALLERTON.2010.5707067}{Pages 1327--1333}.
\newblock IEEE~(2010).

\bibitem{1023345}
T.~Ogawa and H.~Nagaoka.
\newblock ``A new proof of the channel coding theorem via hypothesis testing in quantum information theory''.
\newblock In Proceedings IEEE International Symposium on Information Theory.
\newblock \href{https://dx.doi.org/10.1109/ISIT.2002.1023345}{Page~73}.
\newblock ~(2002).
\newblock  \href{http://arxiv.org/abs/quant-ph/0208139}{arXiv:quant-ph/0208139}.

\bibitem{SN96}
Benjamin Schumacher and M.~A. Nielsen.
\newblock ``Quantum data processing and error correction''.
\newblock \href{https://dx.doi.org/10.1103/PhysRevA.54.2629}{Physical Review A {\bf 54}, 2629--2635}~(1996).
\newblock  \href{http://arxiv.org/abs/quant-ph/9604022}{arXiv:quant-ph/9604022}.

\bibitem{Neumann27}
Johann von Neumann.
\newblock ``Thermodynamik quantenmechanischer gesamtheiten''.
\newblock Nachrichten von der Gesellschaft der Wissenschaften zu Göttingen, Mathematisch-Physikalische Klasse {\bf 1927}, 273--291~(1927).
\newblock  url:~\url{http://eudml.org/doc/59231}.

\end{thebibliography}

\pagebreak
\onecolumn
\appendix
\large 

\section{Equivalent definitions of probabilistic one-way distillable secret key}

\label{app:prob_dist_key_eq_defn}

In this section, we justify the definition of one-shot probabilistic one-way distillable secret key of a bipartite state given in \eqref{eq:prob_1W_distill_key_defn} by showing an equivalence between \eqref{eq:prob_1W_distill_key_defn} and a more general definition of the one-shot probabilistic one-way distillable secret key that we propose in \eqref{eq:prob-1W-key-gen-def}.

Consider a general probabilistic one-way secret-key distillation protocol in which a one-way LOCC channel $\mathcal{L}^{\to}_{AB\to X_AA'B'}$ acts on a bipartite state $\rho_{AB}$ to establish the following state:
\begin{equation}\label{eq:key_distill_arb_form}
    \mathcal{L}^{\to}\!\left(\rho_{AB}\right) = p[1]_{X_A}\otimes\gamma^k_{A'B'} + (1-p)[0]_{X_A}\otimes\sigma_{A'B'},
\end{equation}
where $\gamma^k_{A'B'}$ is a bipartite private state holding $\log_2 k$ secret key bits and $\sigma_{A'B'}$ is an arbitrary bipartite state. As such, this protocol can be used to distill an expected number $p\log_2 k$ of  secret key bits from a single instance of $\rho_{AB}$. The one-shot probabilistic one-way distillable secret key of the state $\rho_{AB}$ is defined generally as
\begin{equation}
\label{eq:prob-1W-key-gen-def}
        K_{D}^{(1),\to}(\rho_{AB}) \coloneqq
        \sup_{\substack{p \in [0,1], k\in \mathbb{N}\\ \mathcal{L}^{\to} \in \operatorname{1WL}, \gamma^k_{A'B'}}}\left\{
        \begin{array}{c}
             p \log_2 k:\\
             \mathcal{L}^{\to}\!\left(\rho_{AB}\right) =
             p[1]_{X_A}\otimes\gamma^k_{A'B'} \\
             + (1-p)[0]_{X_A}\otimes\sigma_{A'B'},\\
             \sigma_{A'B'}\in \mathcal{S}\!\left(A'B'\right)
        \end{array}
        \right\}.
    \end{equation}

In the presence of forward classical communication, Alice can send a copy of the flag $X_A$ to Bob, so that both parties know when the protocol is unsuccessful in establishing the private state $\gamma^k_{A'B'}$, hence, arriving at the following state:
\begin{equation}
    (\mathcal{C}_{X_A\to X_AX_B}\circ\mathcal{L}^{\to})\left(\rho_{AB}\right) = p[1]_{X_A}\otimes[1]_{X_B}\otimes\gamma^k_{A'B'}
    + (1-p)[0]_{X_A}\otimes[0]_{X_B}\otimes\sigma_{A'B'},
\end{equation}
where $\mathcal{C}_{X_A\to X_AX_B}$ is a classical channel that copies the classical data from $X_A$ to $X_B$. Now that both parties hold a copy of the flag, they can trace out their states and replace them with a locally prepared maximally mixed state if the flag indicates that the protocol was unsuccessful in establishing a private state. That is, Alice applies the following local channel on her system:
\begin{equation}
    \mathcal{R}_{X_AA'\to X_AA'}\!\left(\cdot\right) = \frac{I_{A'}}{d_{A'}}\otimes\operatorname{Tr}_{A'}\!\left[[0]_{X_A}\left(\cdot\right)[0]_{X_A}\right] + \operatorname{id}_{X_AA'\to X_AA'}\!\left([1]_{X_A}\left(\cdot\right)[1]_{X_A}\right),
\end{equation}
and Bob applies the corresponding local channel on his systems. Tracing out Bob's flag, we arrive at the following quantum state:
\begin{align}
    (\mathcal{R}\circ\mathcal{C}\circ\mathcal{L}^{\to})\!\left(\rho_{AB}\right)
    &= p[1]_{X_A}\otimes\gamma^k_{A'B'} + (1-p)[0]_{X_A}\otimes\frac{I_{A'}}{d_{A'}}\otimes\frac{I_{B'}}{d_{B'}}\\
    &= p[1]_{X_A}\otimes\gamma^k_{A'B'} + (1-p)[0]_{X_A}\otimes\frac{I_{A'B'}}{d_{A'}d_{B'}}.\label{eq:key_distill_mix_form}
\end{align}

Observe that $\mathcal{R}\circ\mathcal{C}\circ\mathcal{L}^{\to}$ is also a one-way LOCC channel that distills an expected number $p\log_2 k$ of  secret key bits from a single instance of $\rho_{AB}$. Since such a one-way LOCC transformation can be designed for every one-way LOCC channel $\mathcal{L}^{\to}_{AB\to X_AA'B'}$, we can conclude the following statement: if there exists a one-way LOCC channel that distills a quantum state of the form given in \eqref{eq:key_distill_arb_form}, then there also exists a one-way LOCC channel that distills the quantum state of the form given in \eqref{eq:key_distill_mix_form}. As the expected number of secret key bits distilled from $\rho_{AB}$ using either of the channels is equal to $p\log_2 k$, we can restrict the optimization in the definition of one-shot probabilistic one-way distillable secret key to channels that establish a quantum state of the form given in \eqref{eq:key_distill_mix_form}. Hence, we arrive at the equivalent definition of the one-shot probabilistic one-way distillable secret key given in \eqref{eq:prob_1W_distill_key_defn}.

\section{The min-unextendible entanglement of bipartite states}\label{app:min_unext_ent}

A measure for quantifying the unextendibility of quantum states, called unextendible entanglement, was introduced in~\cite{WWW21}. Let us briefly discuss this quantity, which is one of the major components used to prove our main results. Recall that a generalized divergence~$\mathbf{D}$ is a function of two quantum states that is non-increasing under the action of a quantum channel~\cite{polyanskiy2010arimoto}. The unextendible entanglement of a bipartite  state~\cite{WWW21} is defined in terms of a generalized divergence as follows:

\begin{definition}[\cite{WWW21}]
\label{def:unext-ent}
    The generalized unextendible entanglement of a bipartite state $\rho_{AB}$, induced by a generalized divergence~$\mathbf{D}$, is defined as
        \begin{equation}
        \label{eq:gen_unext_ent_states}
        \mathbf{E}^u(\rho_{AB}) \coloneqq  \inf_{\sigma_{AB}\in \mathcal{F}\left(\rho_{AB}\right)} \frac{1}{2} \mathbf{D}\!\left(\rho_{AB}\Vert\sigma_{AB}\right),
    \end{equation}    
    where $\mathcal{F}\!\left(\rho_{AB}\right)$ is defined in \eqref{eq:extension_set}. 
\end{definition}

We also use the notation $\mathbf{E}^u\!\left(\rho_{A:B}\right)$ to clarify the bipartition of systems over which the unextendible entanglement is being considered.

The unextendible entanglement provides a framework for quantifying the unextendibility of a bipartite state $\rho_{AB}$ with respect to the system $B$. Crucially, this quantity is monotonic under the action of two-extendible channels~\cite[Theorem~2]{WWW21}, and hence, it is also monotonic under the action of one-way LOCC channels. (See~\cite{KDWW19,KDWW21,WWW21} for the definition of a two-extendible channel.) Stated formally, \cite[Theorem~2]{WWW21} establishes the following: for every quantum state $\rho_{AB}$ and two-extendible channel $\mathcal{N}_{AB\to A'B'}$, the following inequality holds:
    \begin{equation}
        \mathbf{E}^u(\rho_{AB}) \ge \mathbf{E}^u(\mathcal{N}_{AB\to A'B'}\!\left(\rho_{AB}\right)) .
        \label{eq:unext-ent-monotone-WWW21}
    \end{equation}

A different measure for unextendibility was proposed in~\cite{KDWW19, KDWW21} where the divergence was measured with respect to a fixed set of two-extendible states. However, in Definition~\ref{def:unext-ent}, the divergence is measured by means of a set of states that depend on the input state itself. Although both measures are equal to the minimal possible value of $\mathbf{D}$ when $\rho_{AB}$ is two-extendible, they are not equal in general.

A particular example of interest and one of the main tools in our paper is the min-unextendible entanglement~\cite{WWW21}, defined as follows:
\begin{equation}
    E^u_{\operatorname{min}}\!\left(\rho_{AB}\right) \coloneqq \inf_{\sigma_{AB}\in \mathcal{F}\left(\rho_{AB}\right)}\frac{1}{2} D_{\min}\!\left(\rho_{AB}\Vert\sigma_{AB}\right),
\end{equation}
where $D_{\min}$ is the min-relative entropy~\cite[Def.~2]{Dat09}:
\begin{equation}\label{eq:D_min_proj_eq}
    D_{\min}\!\left(\omega\Vert\tau\right) \coloneqq  -\log_2\operatorname{Tr}\!\left[\Pi^{\omega}\tau\right],
\end{equation}
with $\omega$ a state, $\tau$ a positive semidefinite operator, and $\Pi^{\omega}$ the projection onto the support of $\omega$.

The min-unextendible entanglement has the following properties~\cite{WWW21}:
\begin{enumerate}
    \item \textbf{Non-negativity:} The min-unextendible entanglement is non-negative for a bipartite quantum state, as a consequence of the min-relative entropy being non-negative for all pairs of states. That is,
    \begin{equation}
        E^u_{\min}\!\left(\rho_{AB}\right) \ge 0 \quad \forall \rho_{AB}\in \mathcal{S}\!\left(AB\right).
        \label{eq:min-unext-non-neg}
    \end{equation}
    \item \textbf{Additivity~\cite[Prop.~15]{WWW21}:} The min-unextendible entanglement is additive with respect to a tensor product of the states $\rho_{A_1B_1}$ and $\sigma_{A_2 B_2}$:
    \begin{equation}
        E^u_{\operatorname{min}}\!\left(\rho_{A_1:B_1}\otimes \sigma_{A_2:B_2}\right) = E^u_{\operatorname{min}}\!\left(\rho_{A_1:B_1}\right) + E^u_{\operatorname{min}}\!\left(\sigma_{A_2:B_2}\right).
        \label{eq:min-unext-additive}
    \end{equation}
    \item \textbf{Privacy bound~\cite[Prop.~21]{WWW21}} The min-unextendible entanglement of a bipartite private state $\gamma^k_{AB}$, holding $\log_2 k$ secret key bits, is bounded from below as follows:
    \begin{equation}
        E^u_{\min}\!\left(\gamma^k_{AB}\right) \ge \log_2 k.
    \end{equation}
\end{enumerate}

\section{Proof of Lemma~\ref{lem:min_unext_ent_2_eras_eq}}\label{app:min_unext_ent_2_eras_eq}

We restate Lemma~\ref{lem:min_unext_ent_2_eras_eq} here for the reader's convenience.

\emph{
    For all $p \in [0,1]$ and every integer $k \geq 2$, the min-unextendible entanglement of a doubly erased private state $\eta^{p,k}_{AB}$ is bounded from below by the following quantity:}
    \begin{equation}
        E^u_{\min}\!\left(\eta^{p,k}_{AB}\right) \ge -\frac{1}{2}\log_2\!\left(\frac{p}{k^2} + 1-p\right)  .
    \end{equation}

\medskip

Consider the following partially dephasing channel:
\begin{equation}\label{eq:dephas_channel_defn}
        \Delta_{A}\!\left(\cdot\right) \coloneqq \Pi_A\left(\cdot\right)\Pi_A + [e]_A\left(\cdot\right)[e]_A,
\end{equation}
where $\Pi_A \coloneqq \sum_{i=0}^{d-1} [i]$, so that $\Pi_A [e] = [e]\Pi_A = 0$.
In Lemmas~\ref{lem:dephased_ext} and~\ref{lem:dephas_ext_form}, 
we first characterize the action of $\Delta_A$ on an arbitrary state $\omega_{AE} \in \mathcal{F}\!\left(\eta^{p,k}_{AB}\right)$. We then use  Lemmas~\ref{lem:dephased_ext} and~\ref{lem:dephas_ext_form} to prove Lemma~\ref{lem:min_unext_ent_2_eras_eq}.

\begin{lemma}\label{lem:dephased_ext}
    Let $\omega_{AE}$ be an arbitrary state in the set $\mathcal{F}\!\left(\eta^{p,k}_{AB}\right)$, where $\eta^{p,k}$ is a doubly erased state, as defined in \eqref{eq:double_eras_priv_st}. Then $\Delta_A\!\left(\omega_{AE}\right)$  is also in the set $\mathcal{F}\!\left(\eta^{p,k}_{AB}\right)$, where $\Delta_A$ is defined in~\eqref{eq:dephas_channel_defn}.
\end{lemma}

\begin{proof}
    For every quantum state $\omega_{AE} \in \mathcal{F}\!\left(\eta^{p,k}_{AB}\right)$, there exists a state $\omega_{ABE}$ with the following marginals;
    \begin{align}
        \operatorname{Tr}_E\!\left[\omega_{ABE}\right] &= \eta^{p,k}_{AB},\\
        \operatorname{Tr}_B\!\left[\omega_{ABE}\right] &= \omega_{AE}.
    \end{align}

    Note that the doubly erased state $\eta^{p,k}_{AB}$ is invariant under the action of $\Delta_A$. Therefore, the following equalities hold:
    \begin{align}
        \Delta_A\!\left(\eta^{p,k}_{AB}\right) &= \Delta_A\!\left(\operatorname{Tr}_E\!\left[\omega_{ABE}\right]\right)\\
        &= \operatorname{Tr}_E\!\left[\Delta_A\!\left(\omega_{ABE}\right)\right]\\
        &= \eta^{p,k}_{AB}.
    \end{align}
    This implies that $\Delta\!\left(\omega_{ABE}\right)$ is also an extension of $\eta^{p,k}_{AB}$, and consequently, $\operatorname{Tr}_B\!\left[\Delta_A\!\left(\omega_{ABE}\right)\right] = \Delta_A\!\left(\omega_{AE}\right)$ is in the set~$\mathcal{F}\!\left(\eta^{p,k}_{AB}\right)$.
\end{proof}

\begin{lemma}\label{lem:dephas_ext_form}
    Let $\omega_{AE}$ be an arbitrary state in the set $\mathcal{F}\!\left(\eta^{p,k}_{AB}\right)$, where $\eta^{p,k}$ is a doubly erased state of the following form:
    \begin{equation}
        \eta^{p,k}_{AB} = p~\gamma^k_{AB} + (1-p)[e]_A\otimes[e]_B,
    \end{equation}
    with $\gamma^k_{AB}$ a bipartite private state holding $\log_2 k$ secret key bits. Here systems $E$ and $B$ are isomorphic to each other.
    The action of the partially dephasing channel $\Delta_A$, defined in \eqref{eq:dephas_channel_defn}, on the state $\omega_{AE}$ results in the following state:
    \begin{equation}
    \Delta_A\!\left(\omega_{AE}\right) = p~\sigma_{AE} + (1-p)[e]_A\otimes\tau_E,
    \end{equation}
    where $\sigma_{AE} \in \mathcal{F}\!\left(\gamma^k_{AB}\right)$ and $\tau_E \in \mathcal{S}\!\left(E\right)$.
\end{lemma}

\begin{proof}
    For every quantum state $\omega_{AE} \in \mathcal{F}\!\left(\eta^{p,k}_{AB}\right)$, there exists a state $\omega_{ABE}$ with the following marginals;
    \begin{align}
        \operatorname{Tr}_E\!\left[\omega_{ABE}\right] &= \eta^{p,k}_{AB},\\
        \operatorname{Tr}_B\!\left[\omega_{ABE}\right] &= \omega_{AE}.
    \end{align}

    Recall the definition of the partially dephasing channel $\Delta_A$ from \eqref{eq:dephas_channel_defn}. 
    The state arising from the action of  $\Delta_A$ on the state $\omega_{ABE}$ can be decomposed as follows:
    \begin{equation}\label{eq:dephas_decomp}
        \Delta_{A}\!\left(\omega_{ABE}\right) = X_{ABE} + [e]_{A}\otimes Y_{BE},
    \end{equation}
    where $[e]_AX_{ABE} = X_{ABE}[e]_A = 0$. Since $\Delta_{A}\!\left(\omega_{ABE}\right)$ is a quantum state, $\Delta\!\left(\omega_{ABE}\right) \ge 0$, which implies the following operator inequalities:
    \begin{align}
        X_{ABE} &\ge 0,\label{eq:X_pos}\\
        [e]_A\otimes Y_{BE} &\ge 0,\\
        \Rightarrow Y_{BE} &\ge 0\label{eq:Y_pos},
    \end{align}
    where we have used the fact that $X_{ABE}$ and $[e]_A\otimes Y_{BE}$ are orthogonal to each other. 
    
    Consider the marginal of $\Delta_{A }\!\left(\omega_{ABE}\right)$ on systems $AB$, which is equal to the following:
    \begin{align}
        \operatorname{Tr}_{E}\!\left[\Delta_{A}\!\left(\omega_{ABE}\right)\right] &=\operatorname{Tr}_{E}\!\left[X_{ABE}\right] + \operatorname{Tr}_{E}\!\left[Y_{BE}\right]\otimes[e]_A.\label{eq:marg_A_decomp}
    \end{align}
    We can also evaluate the marginal by first tracing out the system~$E$ and then applying the partially dephasing channel; that is,
    \begin{align}
        \operatorname{Tr}_{E}\!\left[\Delta_{A}\!\left(\omega_{ABE}\right)\right] &= \Delta_A\left(\operatorname{Tr}_{E}\!\left[\omega_{ABE}\right]\right)\\
        &= p~\gamma^k_{AB} + (1-p)[e]_A\otimes[e]_B.\label{eq:marg_A_orig}
    \end{align}
    Comparing \eqref{eq:marg_A_decomp} and \eqref{eq:marg_A_orig}, and using the fact that $[e]_AX_{ABE} = 0$, we arrive at the following equalities:
    \begin{align}
        \operatorname{Tr}_{E}\!\left[X_{ABE}\right] &= p~\gamma^k_{AB},\label{eq:X_marg}\\
        \operatorname{Tr}_{E}\!\left[Y_{BE}\right] &= (1-p)[e]_B.\label{eq:Y_marg}
    \end{align}
    
    Recall from \eqref{eq:X_pos} and \eqref{eq:Y_pos} that $X_{ABE}$ and $Y_{BE}$ are positive semidefinite operators. Using \eqref{eq:X_marg}, we can further conclude that $\frac{X_{ABE}}{p}$ is a quantum state that extends $\gamma^k_{AB}$. Similarly, \eqref{eq:Y_marg} implies that $\frac{Y_{BE}}{1-p}$ is a quantum state that extends~$[e]_B$. 
    
    Any quantum state that extends $[e]_B$ is of the form $[e]_B\otimes\tau_E$, where $\tau_E$ is an arbitrary quantum state (and which can even be~$[e]_E$). Let us define the quantum state $\sigma_{AE} \coloneqq \frac{1}{p}\operatorname{Tr}_B\!\left[X_{ABE}\right]$. Since $\frac{1}{p}\operatorname{Tr}_E\!\left[X_{ABE}\right] = \gamma^k_{AB}$, $\sigma_{AE}$ is in the set $\mathcal{F}\!\left(\gamma^k_{AB}\right)$. The marginal of the state $\Delta_A\!\left(\omega_{ABE}\right)$ on systems $AE$ can now be written as follows:
    \begin{equation}
        \Delta_A\!\left(\omega_{AE}\right) = p~\sigma_{AE} + (1-p)[e]_A\otimes\tau_E,
    \end{equation}
    where $\tau_E \in \mathcal{S}\!\left(E\right)$ and $\sigma_{AE} \in \mathcal{F}\!\left(\gamma^k_{AB}\right)$.
\end{proof}

\medskip

\begin{proof}[Proof of Lemma~\ref{lem:min_unext_ent_2_eras_eq}]
The min-unextendible entanglement of $\eta^{p,k}_{AB}$ is equal to the following:
\begin{equation}\label{eq:min_unext_ent_2_eras_priv_st_stmnt}
    E^u_{\min}\!\left(\eta^{p,k}_{AB}\right) =  \inf_{\omega_{AB} \in \mathcal{F}\left(\eta^{p,k}\right)}\frac{1}{2}D_{\min}\!\left(\eta^{p,k}_{AB}\Big\Vert\omega_{AB}\right),
\end{equation}
where $\eta^{p,k}_{AB} = p~\gamma^k_{AB} + (1-p)[e]_A\otimes[e]_B$ with $\gamma^k_{AB}$ a bipartite private state holding $\log_2 k$ secret key bits.

Consider an arbitrary state $\omega_{AB} \in \mathcal{F}\!\left(\eta^{p,k}_{AB}\right)$. The min-relative entropy between $\eta^{p,k}_{AB}$ and $\omega_{AB}$ obeys the following inequality:
\begin{align}
    D_{\min}\!\left(\eta^{p,k}_{AB}\Big\Vert\omega_{AB}\right) &\ge D_{\min}\!\left(\Delta_A\!\left(\eta^{p,k}_{AB}\right)\Big\Vert\Delta_A\!\left(\omega_{AB}\right)\right)
    \label{eq:min_rel_ent_dec_delta_other}\\
    &= D_{\min}\!\left(\eta^{p,k}_{AB}\Big\Vert\Delta_A\!\left(\omega_{AB}\right)\right),\label{eq:min_rel_ent_dec_delta}
\end{align}
where $\Delta_A$ is defined in \eqref{eq:dephas_channel_defn}. The above inequality follows from the data-processing of min-relative entropy and the subsequent equality follows from the invariance of the doubly erased state under the action of the partially dephasing channel $\Delta_A$. We have shown in Lemma~\ref{lem:dephased_ext} that $\Delta_A\!\left(\omega_{AB}\right)$ also is in the set $\mathcal{F}\!\left(\eta^{p,k}_{AB}\right)$. Then the inequality in \eqref{eq:min_rel_ent_dec_delta_other}--\eqref{eq:min_rel_ent_dec_delta} implies that we can restrict the optimization in \eqref{eq:min_unext_ent_2_eras_priv_st_stmnt} to states $\omega_{AB}$ such that $\omega_{AB} = \Delta_A\!\left(\omega_{AB}\right)$. Alternatively, we can write the min-unextendible entanglement of the doubly erased state as follows:
\begin{equation}
    E^u_{\min}\!\left(\eta^{p,k}_{AB}\right) =  \inf_{\omega_{AB} \in \mathcal{F}\left(\eta^{p,k}\right)}\frac{1}{2}D_{\min}\!\left(\eta^{p,k}_{AB}\Big\Vert\Delta_A\!\left(\omega_{AB}\right)\right).
\end{equation}

We have shown in Lemma~\ref{lem:dephas_ext_form} that the action of the channel $\Delta_A$ on every state $\omega_{AB}\in \mathcal{F}\!\left(\eta^{p,k}_{AB}\right)$ results in a state of the following form:
\begin{equation}
    \Delta_A\!\left(\omega_{AB}\right) = p~\sigma_{AB} + (1-p)[e]_A\otimes\tau_B,
\end{equation}
where $\tau_B \in \mathcal{S}\!\left(B\right)$, and $\sigma_{AB} \in \mathcal{F}\!\left(\gamma^k_{AB}\right)$. This leads to the following equalities:
\begin{align}
    E^u_{\min}\!\left(\eta^{p,k}_{AB}\right)
    &=  \inf_{\omega_{AB} \in \mathcal{F}\left(\eta^{p,k}\right)}\frac{1}{2}D_{\min}\!\left(\eta^{p,k}_{AB}\Big\Vert\Delta_A\!\left(\omega_{AB}\right)\right)\\
    &=  \inf_{\sigma, \tau}- \frac{1}{2}\log_2 \operatorname{Tr}\!\left[\Pi^{\eta^{p,k}}\left(p~\sigma_{AB} + (1-p)[e]_{A}\otimes \tau_{B}\right)\right]\\
    &=  - \frac{1}{2}\log_2 \sup_{\sigma, \tau}\operatorname{Tr}\!\left[\Pi^{\eta^{p,k}}\left(p~\sigma_{AB} + (1-p)[e]_{A}\otimes \tau_{B}\right)\right],
\end{align}
where $\Pi^{\eta^{p,k}}$ is the projection onto the support of $\eta^{p,k}_{AB}$. The optimization in the second and third equalities is over all $\sigma_{AB} \in \mathcal{F}\!\left(\gamma^k_{AB}\right)$ and $\tau_{B}\in \mathcal{S}\!\left(B\right)$. The last equality follows from the monotonicity of the logarithm. 

Note that $\gamma^k_{AB}$ is orthogonal to $[e]_{A}\otimes[e]_{B}$, and $[e]_{A}\otimes[e]_{B}$ is a projection. Therefore, we can write the projection onto the support of $\eta^{p,k}_{AB}$ as the following sum:
\begin{equation}
    \Pi^{\eta^{p,k}}_{AB} = \Pi^{\gamma^{k}}_{AB} + [e]_{A}\otimes[e]_{B},
\end{equation}
where $\Pi^{\gamma^{k}}_{AB}$ is the projection onto the support of the bipartite private state $\gamma^k_{AB}$. Consequently,
\begin{align}
    E^u_{\min}\!\left(\eta^{p,k}_{AB}\right)
    &= -\frac{1}{2}\log_2 \sup_{\sigma,\tau} \left(p\operatorname{Tr}\!\left[\Pi^{\gamma^k}_{AB}\sigma_{AB}\right] 
    + (1-p)\operatorname{Tr}\!\left[\left([e]_{A}\otimes[e]_{B}\right)\left([e]_{A}\otimes\tau_{B}\right)\right]\right)\\
    &= -\frac{1}{2}\log_2 \sup_{\sigma,\tau} \left(p\operatorname{Tr}\!\left[\Pi^{\gamma^k}_{AB}\sigma_{AB}\right]+ (1-p)\operatorname{Tr}\!\left[[e]_{A}\otimes\left([e]_{B}\tau_{B}\right)\right]\right),\label{eq:min_unext_ent_double_eras_priv_expansion}
\end{align}
where the first equality follows from the fact that $\left([e]_A\otimes I_B\right)\sigma_{AB} = 0$ and $\Pi^{\gamma^k}_{AB}\left([e]_{A}\otimes\tau_{B}\right) = 0$ for every quantum state $\tau_{B}$. 

It is clear that $\tau_{B}$ should be $[e]_{B}$ in order for the second term inside the logarithm in \eqref{eq:min_unext_ent_double_eras_priv_expansion} to be non-zero. Therefore,
\begin{align}
    E^u_{\min}\!\left(\eta^{p,k}_{AB}\right)
    &= -\frac{1}{2}\log_2 \sup_{\sigma\in \mathcal{F}\left(\gamma^k_{AB}\right)}\left(p\operatorname{Tr}\!\left[\Pi^{\gamma^k}\sigma_{AB}\right] + 1-p\right)\\
    &= -\frac{1}{2}\log_2 \sup_{\sigma\in \mathcal{F}\left(\gamma^k_{AB}\right)}\left(p~ 2^{-D_{\min}\!\left(\gamma^k_{AB}\Big\Vert\sigma_{AB}\right)} + 1-p\right)\\
    &= -\frac{1}{2}\log_2 \left(p~2^{-\inf_{\sigma\in \mathcal{F}\left(\gamma^k_{AB}\right)}D_{\min}\!\left(\gamma^k_{AB}\Big\Vert\sigma_{AB}\right)}+ 1-p\right)\\
    &= -\frac{1}{2}\log_2 \left(p~2^{-2E^u_{\min}\!\left(\gamma^k_{AB}\right)}+ 1-p\right),
\end{align}
where the second equality follows from the definition of min-relative entropy, the third equality follows from the monotonicity of the exponential function, and the last equality follows from the definition of the min-unextendible entanglement of quantum states.

It has been shown in~\cite[Corollary 22]{WWW21} that the min-unextendible entanglement of a bipartite private state holding $\log_2 k$ bits of secrecy is not less than $\log_2 k$. That is,
\begin{equation}
    E^u_{\min}\!\left(\gamma^k_{AB}\right) \ge \log_2 k.
\end{equation}
Therefore,
\begin{equation}
    p~2^{-2E^u_{\min}\!\left(\gamma^k\right)} \le p~2^{-2\log_2 k} = \frac{p}{k^2}.
\end{equation}
From the monotonicity of the logarithm function, 
\begin{equation}
    \log_2 \left(p~2^{-2E^u_{\min}\!\left(\gamma^k_{AB}\right)}+ 1-p\right)
    \le \log_2\left(\frac{p}{k^2} + 1-p\right),
\end{equation}
and hence,
\begin{equation}
    E^u_{\min}\!\left(\eta^{p,k}_{AB}\right) \ge -\frac{1}{2}\log_2\left(\frac{p}{k^2} + 1-p\right) \quad \forall p\in (0,1), k\ge 2.
\end{equation}

When $p=0$, $\eta^{p,k}_{AB}$ is a separable state, and hence, its min-unextendible entanglement is equal to zero \cite[Proposition~3]{WWW21}. When $p=1$, $\eta^{p,k}_{AB}$ is a bipartite private state holding $\log_2 k$ bits of secrecy, and hence, its min-unextendible entanglement is not less than $\log_2 k$ as stated in~\cite[Corollary 22]{WWW21}. Thus, we conclude the statement of Lemma~\ref{lem:min_unext_ent_2_eras_eq}.
\end{proof}

\section{The set of super two-extendible states is convex but not closed}

\label{app:sup_2_ext_convex_open}

First we show that the set of super two-extendible states is convex. Consider two quantum states $\rho_{AB}$ and $\tau_{AB}$ that are in the set of super two-extendible states. Let $\overline{\rho}_{ABE}$ be an extension of $\rho_{AB}$ such that the following containment holds:
\begin{equation}\label{eq:rho_contained}
    \operatorname{supp}\!\left(\operatorname{id}_{E \to B}(\operatorname{Tr}_B\!\left[\overline{\rho}_{ABE}\right])\right) \subseteq \operatorname{supp}\!\left(\rho_{AB}\right),
\end{equation}
where both systems $E$ and $B$ are modeled by the same Hilbert space. Note that the existence of such an extension is guaranteed by the fact that $\rho_{AB}$ is in the set of super two-extendible states. Similarly, let $\overline{\tau}_{ABE}$ be an extension of $\tau_{AB}$ such that the following containment holds:
\begin{equation}\label{eq:tau_contained}
    \operatorname{supp}\!\left(\operatorname{id}_{E \to B}(\operatorname{Tr}_B\!\left[\overline{\tau}_{ABE}\right])\right) \subseteq \operatorname{supp}\!\left(\tau_{AB}\right).
\end{equation}
From now on, we will assume that the identity channel $\operatorname{id}_{E\to B}$ is implicit when we compare the supports of two operators in $\mathcal{S}(AB)$ and $\mathcal{S}(AE)$, respectively.

Consider an arbitrary convex combination of the states $\rho_{AB}$ and $\tau_{AB}$:
\begin{equation}
    \sigma_{AB} \coloneqq p~\rho_{AB} + (1-p)\tau_{AB},
\end{equation}
where $p \in [0,1]$.
The following state is a valid extension of~$\sigma_{AB}$:
\begin{equation}
    \omega_{ABE} \coloneqq p~\overline{\rho}_{ABE} + (1-p)\overline{\tau}_{ABE},
\end{equation}
because
\begin{equation}
    \operatorname{Tr}_E\!\left[\omega_{ABE}\right] = p~\rho_{AB} + (1-p)\tau_{AB} = \sigma_{AB}.
\end{equation}
The other relevant marginal of $\omega_{ABE}$ can be expressed as follows:
\begin{equation}
    \operatorname{Tr}_{B}\!\left[\omega_{ABE}\right] = p\operatorname{Tr}_{B}\!\left[\overline{\rho}_{ABE}\right] + (1-p)\operatorname{Tr}_{B}\!\left[\overline{\tau}_{ABE}\right].
\end{equation}

Note that $\sigma_{AB}$ is in the set of super two-extendible states for $p=0$ and $p=1$ by the assumption that $\rho_{AB}$ and $\tau_{AB}$ are super two-extendible states. For all $p\in (0,1)$, the following holds:
\begin{align}
    \operatorname{supp}\!\left(\operatorname{Tr}_B\!\left[\omega_{ABE}\right]\right)
    &= \operatorname{supp}\!\left(\operatorname{Tr}_B\!\left[\overline{\rho}_{ABE}\right]\right)\cup\operatorname{supp}\!\left(\operatorname{Tr}_B\!\left[\overline{\tau}_{ABE}\right]\right)\\
    &\subseteq \operatorname{supp}\!\left(\operatorname{Tr}_E\!\left[\overline{\rho}_{ABE}\right]\right)\cup\operatorname{supp}\!\left(\operatorname{Tr}_E\!\left[\overline{\tau}_{ABE}\right]\right)\\
    &= \operatorname{supp}\!\left(\operatorname{Tr}_E\!\left[\omega_{ABE}\right]\right)\\
    &= \operatorname{supp}\!\left(\sigma_{AB}\right),
    \label{eq:omega_contained}
\end{align}
where the containment of the sets follows from \eqref{eq:rho_contained} and \eqref{eq:tau_contained}. 

Since the support of $\operatorname{Tr}_B\!\left[\omega_{ABE}\right]$ is in the support of $\operatorname{Tr}_E\!\left[\omega_{ABE}\right]$, we conclude that $\operatorname{Tr}_E\!\left[\omega_{ABE}\right] = \sigma_{AB}$ is a super two-extendible state for all $p\in (0,1)$ as well. Therefore, the state $\sigma_{AB}$, which is a convex combination of two arbitrary super two-extendible states, is super two-extendible for every $p\in [0,1]$, justifying that the set of super two-extendible states is convex.

To show that the set of super two-extendible states is not closed, let us consider the example of erased states defined in \eqref{eq:eras_st_defn}. The erased state $\widetilde{\Phi}^{p,k}_{AB}$ is a special case of an erased private state $\widetilde{\eta}^{p,k}_{AB}$. Appendix~\ref{app:eras_min_unext_ent_eq_0} establishes that an erased private state is super two-extendible for every $p\in [0,1)$. This implies that an erased state $\widetilde{\Phi}^{p,k}_{AB}$ is also super two-extendible for every $p\in [0,1)$. However, for $p=1$, $\widetilde{\Phi}^{p,k}_{AB}$ is a maximally entangled state of Schmidt rank $k$, which is not super two-extendible. In fact, this state has a min-unextendible entanglement of $\log_2 k$. As such, the maximally entangled state is a limit point of the set of super two-extendible states, but it is not a super two-extendible state itself. Hence, the set of super two-extendible states is not closed.

\section{Proof of Proposition~\ref{prop:ads_min_unext_ent_eq_0}} \label{app:ads_min_unext_ent_eq_0}

We restate Proposition~\ref{prop:ads_min_unext_ent_eq_0} here for the reader's convenience.

\emph{The min-unextendible entanglement of a quantum state is equal to zero if and only if it is super two-extendible. }

\medskip

For every super two-extendible state $\rho_{AB}$, there exists a quantum state $\sigma_{AB}\in \mathcal{F}\!\left(\rho_{AB}\right)$ such that $\operatorname{supp}\!\left(\sigma_{AB}\right)\subseteq \operatorname{supp}\!\left(\rho_{AB}\right)$. This implies that $\Pi^{\rho}_{AB}\sigma_{AB}  = \sigma_{AB}$, where $\Pi^{\rho}_{AB}$ is the projection onto the support of $\rho_{AB}$. The min-unextendible entanglement of $\rho_{AB}$ can then be bounded as follows:
\begin{align}
    0 &\le E^u_{\min}\!\left(\rho_{AB}\right)\\
    &=\inf_{\tau \in \mathcal{F}\left(\rho\right)} -\frac{1}{2}\log_2\operatorname{Tr}\!\left[\Pi^{\rho}_{AB}\tau_{AB}\right]\\
    &\le -\frac{1}{2}\log_2\operatorname{Tr}\!\left[\Pi^{\rho}_{AB}\sigma_{AB}\right]\\
    &= -\frac{1}{2}\log_2\operatorname{Tr}\!\left[\sigma_{AB}\right]\\
    &= 0,
\end{align}
where the first inequality follows from \eqref{eq:min-unext-non-neg}, the first equality follows from the definition of min-unextendible entanglement of a bipartite state, and the last equality follows from the fact that $\sigma_{AB}$ is a quantum state with unit trace. Therefore, the min-unextendible entanglement of a super two-extendible state is equal to zero.

We now establish the opposite implication. Consider an arbitrary quantum state $\rho_{AB}$ such that $E^u_{\min}\!\left(\rho_{AB}\right) = 0$. This is true only if there exists a quantum state $\sigma_{AB}\in \mathcal{F}\!\left(\rho_{AB}\right)$ such that $\operatorname{Tr}\!\left[\Pi^{\rho}_{AB}\sigma_{AB}\right] = 1$, which follows from the definition of the min-unextendible entanglement of a bipartite state. From the gentle measurement lemma~\cite[Lemma~5]{1023345}, it is known for a state $\tau$ and a projector $P$ that
\begin{equation}
    \frac{1}{2} \left \| \tau - P \tau P\right\|_1 \leq \sqrt{1 - \operatorname{Tr}[P \tau]}.
\end{equation}
Applying this to our case, we conclude that 
\begin{equation} 
\left\|\Pi^{\rho}_{AB}\sigma_{AB}  \Pi^{\rho}_{AB} - \sigma_{AB} \right\|_1 = 0,
\end{equation}
which implies that $\Pi^{\rho}_{AB}\sigma_{AB}  \Pi^{\rho}_{AB} = \sigma_{AB}$. This in turns implies that the support of $\sigma_{AB}$ is in the support of $\rho_{AB}$. We finally conclude that the min-unextendible entanglement of $\rho_{AB}$ is equal to zero only if there exists a quantum state $\sigma_{AB}\in \mathcal{F}\!\left(\rho_{AB}\right)$ such that $\operatorname{supp}\!\left(\sigma_{AB}\right) \subseteq \operatorname{supp}\!\left(\rho_{AB}\right)$, which is precisely the definition of a super two-extendible state.

\section{Proof of Theorem~\ref{theo:prob_distill_ads_eq_0}}

\label{app:prob_distill_eras_st_eq_0_proof}

    Let us restate Theorem~\ref{theo:prob_distill_ads_eq_0} for the reader's convenience.
    
    \emph{The probabilistic one-way distillable secret key of a super two-extendible state is equal to zero.}
    
    \medskip

       Consider a one-way LOCC channel $\mathcal{L}^{\to}_{A^nB^n\to X_AA'B'}$ that acts on $n$ copies of a super two-extendible state $\rho_{AB}$ as follows:
    \begin{equation}\label{eq:distill_st_eras_st}
        \mathcal{L}^{\to}\!\left(\left(\rho_{AB}\right)^{\otimes n}\right) = q[1]_{X_A}\otimes \gamma^{k'}_{A'B'} + (1-q)[0]_{X_A}\otimes \frac{I_{A'B'}}{d_{A'}d_{B'}},
    \end{equation}
    for some  $q\in [0,1]$ and integer $k'\geq 2$. Then 
    \begin{align}
        \label{eq:proof-first} 0 & \leq 
        E^u_{\min}\!\left(q[1]_{X_A}\otimes \gamma^{k'}_{A'B'} + (1-q)[0]_{X_A}\otimes \frac{I_{A'B'}}{d_{A'}d_{B'}}\right) \\
        & = E^u_{\min}\!\left(\eta^{q,k'}_{A'B'}\right) \\
        &= E^u_{\min}\!\left(\mathcal{L}^{\to}\!\left(\left(\rho_{AB}\right)^{\otimes n}\right)\right) \\
        &\le E^u_{\min}\!\left(\left(\rho_{AB}\right)^{\otimes n}\right)\\
        &= nE^u_{\min}\!\left(\rho_{AB}\right)\\
        &= 0,
        \label{eq:proof-last} 
    \end{align}
    The first inequality follows from \eqref{eq:min-unext-non-neg}. 
    The first equality follows from  Remark~\ref{rem:min_uenxt_ent_eq_tf} (the min-unextendible entanglement of the state in \eqref{eq:distill_st_eras_st} is the same as the min-unextendible entanglement of the doubly erased private state $\eta^{q,k'}_{A'B'}$).
    The second inequality follows from the monotonicity of min-unextendible entanglement under one-way LOCC channels, the third equality follows from the additivity of min-unextendible entanglement (recall \eqref{eq:min-unext-additive}), and the final equality follows from Proposition~\ref{prop:ads_min_unext_ent_eq_0}.
    
    Inspecting \eqref{eq:proof-first} and \eqref{eq:proof-last}, it follows that $E^u_{\min}\!\left(\eta^{q,k'}_{A'B'}\right) = 0$. By applying Lemma~\ref{lem:min_unext_ent_2_eras_eq}, using the fact that $-\frac{1}{2}\log_2\!\left(\frac{q}{k'^2} + 1-q\right) \geq 0$ for all $q \in [0,1]$ and integer $k' \geq 2$, and solving the equation $0 = -\frac{1}{2}\log_2\!\left(\frac{q}{k'^2} + 1-q\right)$, along with the assumption that $k' \geq 2$, we conclude that $q = 0$. As such, the expected number of secret key bits distilled by the channel $\mathcal{L}^{\to}_{A^nB^n\to X_AA'B'}$ from $n$ copies of $\rho_{AB}$ is equal to zero. This conclusion  holds for every one-way LOCC channel $\mathcal{L}^{\to}_{A^nB^n\to X_AA'B'}$, every super two-extendible state $\rho_{AB}$, and every $n\in \mathbb{N}$. Hence, the one-shot probabilistic one-way distillable secret key of $\left(\rho_{AB}\right)^{\otimes n}$ is equal to zero for all $n\in \mathbb{N}$, and consequently, the probabilistic one-way distillable secret key of every super two-extendible state is equal to zero.

\section{Proof of Proposition~\ref{lem:eras_min_unext_ent_eq_0}}

\label{app:eras_min_unext_ent_eq_0}

Let us restate Proposition~\ref{lem:eras_min_unext_ent_eq_0} for the reader's convenience.

\emph{For all $p \in [0,1)$ and every integer $k \geq 2$, the erased private state $\widetilde{\eta}^{p,k}_{AB}$ is a super two-extendible state and thus has probabilistic one-way distillable secret key equal to zero.}

\medskip

    Consider the following extension of an erased private state~$\widetilde{\eta}^{p,k}_{AB}$:
    \begin{equation}
        \omega_{ABE} \coloneqq p~\gamma^{k}_{AB}\otimes[e]_{E} + (1-p)\gamma^{k}_{AE}\otimes[e]_{B},
    \end{equation}
    where $E$ is isomorphic to $B$. The two relevant marginals of the state $\omega_{ABE}$ are 
    \begin{align}
        \operatorname{Tr}_{E}\!\left[\omega_{ABE}\right] &= p~\gamma^k_{AB} + (1-p)\gamma^{k}_{A}\otimes [e]_{B} = \widetilde{\eta}^{p,k}_{AB},\\
        \operatorname{Tr}_{B}\!\left[\omega_{ABE}\right] &= (1-p)\gamma^k_{AE} + p~\gamma^{k}_{A}\otimes [e]_{E} = \widetilde{\eta}^{1-p,k}_{AE},
    \end{align}
    where $\gamma^k_A \coloneqq \operatorname{Tr}_B\!\left[\gamma^k_{AB}\right] = \operatorname{Tr}_E\!\left[\gamma^k_{AE}\right]$.

    Note that the min-relative entropy between two quantum states, $\rho$ and $\sigma$, is equal to zero if $\operatorname{supp}(\rho) = \operatorname{supp}(\sigma)$. We know that $\operatorname{supp}(\widetilde{\eta}^{p,k}_{AB}) = \operatorname{supp}(\widetilde{\eta}^{1-p,k}_{AB})$ for all $p \in (0,1)$ and every integer $k\ge 2$. Therefore, we arrive at the following relations:
    \begin{equation}
        E^u_{\min}\!\left(\widetilde{\eta}^{p,k}_{AB}\right) \le \frac{1}{2}D_{\min}\!\left(\widetilde{\eta}^{p,k}_{AB}\Big\Vert \widetilde{\eta}^{1-p,k}_{AB}\right) = 0 \quad \forall p\in (0,1), k\ge 2.
    \end{equation}
    The min-unextendible entanglement of a bipartite state is a non-negative quantity since the underlying divergence is non-negative. Therefore,
    \begin{equation}
        E^u_{\min}\!\left(\widetilde{\eta}^{p,k}_{AB}\right) = 0 \quad \forall p\in (0,1), k\ge 2.
    \end{equation}

    Finally, the erased private state $\widetilde{\eta}^{p,k}_{AB}$ is a separable state for $p=0$, and the min-unextendible entanglement of a separable state is equal to zero \cite[Proposition~3]{WWW21}. Thus, we conclude the statement of Proposition~\ref{lem:eras_min_unext_ent_eq_0}.

\section{Proof of Corollary~\ref{cor:dist_full_rank}}\label{app:dist_full_rank}

Let us restate Corollary~\ref{cor:dist_full_rank} for the reader's convenience.

\emph{All full-rank states are super two-extendible
states, and the probabilistic one-way distillable secret
key of such quantum states is equal to zero.}

\medskip

Consider a bipartite quantum state that has a full-rank density operator $\rho_{AB}$. The quantum state $\rho_{AB}\otimes\frac{I_{E}}{d_E}$ is a valid extension of the state, and its marginal $\frac{1}{d_E}\operatorname{Tr}_{B}\!\left[\rho_{AB}\otimes I_E\right]$ is in the set $\mathcal{F}\!\left(\rho_{AB}\right)$ when $E\cong B$. The min-unextendible entanglement of the state $\rho_{AB}$ obeys the following inequality by definition:
    \begin{align}
        E^u_{\min}\!\left(\rho_{AB}\right) &\le  \frac{1}{2}D_{\min}\!\left(\rho_{AB}\Big\Vert\frac{1}{d_E}\operatorname{Tr}_{B}\!\left[\rho_{AB}\otimes I_E\right]\right)\\
        &= \frac{1}{2}D_{\min}\!\left(\rho_{AB}\Big\Vert\rho_A\otimes \frac{I_E}{d_E}\right)\\
        &= -\frac{1}{2}\log_2\!\left(\operatorname{Tr}\!\left[\Pi^\rho_{AB}\left(\rho_{A}\otimes \frac{I_{E}}{d_E}\right)\right]\right),
    \end{align}
    where $\Pi^\rho_{AB}$ is the projection onto the support of $\rho_{AB}$, and $\rho_A \coloneqq \operatorname{Tr}_B\!\left[\rho_{AB}\right]$. Since $\rho_{AB}$ is a full-rank density operator, the projection onto the support of $\rho_{AB}$ is the identity operator; that is,
    \begin{equation}
        \Pi^\rho_{AB} = I_{AB}.
    \end{equation}
    Therefore,
    \begin{align}
        0 &\le E^u_{\min}\!\left(\rho_{AB}\right)\label{eq:min_unext_ent_full_rank_ge_0}\\
        &\le -\frac{1}{2}\log_2\!\left(\frac{1}{d_E}\operatorname{Tr}\!\left[I_{AB}\left(\rho_A\otimes I_{E}\right)\right]\right)\\
        &= -\frac{1}{2}\log_2\!\left(\operatorname{Tr}\!\left[\rho_A\right]\right)\\
        &= 0, \label{eq:min_unext_ent_full_rank_le_0}
    \end{align}
    where the first inequality follows from the non-negativity of the min-unextendible entanglement of states. The first equality follows from the the fact that $E\cong B$, and the last equality follows from the fact that $\rho_{A}$ is a quantum state with unit trace. Since \eqref{eq:min_unext_ent_full_rank_ge_0} and \eqref{eq:min_unext_ent_full_rank_le_0} hold for every full-rank state $\rho_{AB}$, we conclude that the min-unextendible entanglement of all full-rank states is equal to zero; that is, all full-rank states are super two-extendible. 
    
    As a consequence of Theorem~\ref{theo:prob_distill_ads_eq_0}, all full-rank states being contained in the set of super two-extendible states implies that the probabilistic one-way distillable secret key of all full-rank states is also equal to zero.

\section{Coherent information of Werner states}\label{app:coherent_info_Werner}

    \begin{proposition}
        Consider a $d$-dimensional Werner state defined as follows:
        \begin{equation}\label{eq:werner_st_full_defn}
            W^{p,d}_{AB}\coloneqq p\frac{I_{AB} + F_{AB}}{d(d+1)} + (1-p)\frac{I_{AB}-F_{AB}}{d(d-1)},
        \end{equation}
        where $p\in [0,1]$, $d = d_A = d_B \in \{2, 3, 4, \ldots \}$, and $F_{AB}$ is the swap operator defined as $F_{AB} \coloneqq \sum_{i,j=0}^{d-1}|i\rangle\!\langle j|_A\otimes |j\rangle\!\langle i|_B$. The coherent information of the state $W^{p,d}_{AB}$ is equal to the following quantity:
        \begin{equation}\label{eq:coherent_info_werner}
		I\!\left(A\rangle B\right)_{W^{p,d}} = 1 - h_2(p) 
 - p\log_2\!\left({d+1}\right) - (1-p)\log_2\!\left({d-1}\right).
	\end{equation}
    \end{proposition}

    \begin{proof}
        The coherent information of a bipartite state $\rho_{AB}$ is defined as follows~\cite{SN96}:
	   \begin{equation}\label{eq:coherent_info_defn}
		  I\!\left(A\rangle B\right)_{\rho} \coloneqq H\!\left(\operatorname{Tr}_A\!\left[\rho_{AB}\right]\right) - H\!\left(\rho_{AB}\right),
	   \end{equation}
	   where $H\!\left(\rho\right)$ is the von Neumann entropy~\cite{Neumann27}, defined as
	   \begin{equation}
		  H\!\left(\rho\right) \coloneqq -\operatorname{Tr}\!\left[\rho\log_2\rho\right].
	   \end{equation} 
        
        The $d$-dimensional Werner state  defined in \eqref{eq:werner_st_full_defn} can be rearranged into the following form:
	\begin{equation}
		W^{p,d}_{AB} = p\frac{2}{d(d+1)}\Pi^{\operatorname{sym}}_{AB} + (1-p)\frac{2}{d(d-1)}\Pi^{\operatorname{asym}}_{AB},
	\end{equation}
	where $\Pi^{\operatorname{sym}}$ and $\Pi^{\operatorname{asym}}$ are the projections onto the symmetric and asymmetric subspaces of the underlying Hilbert space, respectively. This is easily seen from the fact that the symmetric and asymmetric projection operators for a bipartite system can be written as follows:
	\begin{align}
		\Pi^{\operatorname{sym}}_{AB} &= \frac{1}{2}\left(I_{AB} + F_{AB}\right),\label{eq:sym_proj_IF}\\
		\Pi^{\operatorname{asym}}_{AB} &= \frac{1}{2}\left(I_{AB} - F_{AB}\right),\label{eq:asym_proj_IF}
	\end{align}
	where $F_{AB}$ is the swap operator. The symmetric and asymmetric projection operators are orthogonal to each other:
	\begin{equation}
		\Pi^{\operatorname{sym}}_{AB}\Pi^{\operatorname{asym}}_{AB} = 0.
	\end{equation}
	
	Let us first evaluate the von Neumann entropy of $\operatorname{Tr}_A\!\left[W^{p,d}_{AB}\right]$. Note that 
	\begin{equation}
		\operatorname{Tr}_{A}\!\left[F_{AB}\right] = I_{B}.
	\end{equation}
	Therefore,
		\begin{align}
			\operatorname{Tr}_A\!\left[\Pi^{\operatorname{sym}}_{AB}\right] &= \frac{d+1}{2}I_{B},\label{eq:part_tr_sym_op}\\
			\operatorname{Tr}_A\!\left[\Pi^{\operatorname{asym}}_{AB}\right] &= \frac{d-1}{2}I_{B},\label{eq:part_tr_asym_op}
		\end{align}
		and consequently,
		\begin{equation}
			\operatorname{Tr}_{A}\!\left[W^{p,d}_{AB}\right] = \frac{p}{d}I_{B} + \frac{1-p}{d}I_{B} = \frac{I_B}{d},
		\end{equation}
		which is the maximally mixed state on system $B$. The von Neumann entropy of a $d$-dimensional maximally mixed state is equal to $\log_2 d$. Therefore,
		\begin{equation}\label{eq:VN_ent_werner_partial}
			H\!\left(\operatorname{Tr}_A\!\left[W^{p,d}_{AB}\right]\right) = \log_2 d.
		\end{equation}
	
	Now let us evaluate the von Neumann entropy of the Werner state $W^{p,d}_{AB}$:
	\begin{equation}
		H\!\left(W^{p,d}_{AB}\right) = -\operatorname{Tr}\!\left[W^{p,d}_{AB}\log_2 W^{p,d}_{AB}\right].
	\end{equation}
	Since $\Pi^{\operatorname{sym}}_{AB}$ and $\Pi^{\operatorname{asym}}_{AB}$ are orthogonal to each other, we can reduce the von Neumann entropy of the Werner state into the following form:
	\begin{multline}
		H\!\left(W^{p,d}_{AB}\right) = -\operatorname{Tr}\!\left[p\frac{2}{d(d+1)}\Pi^{\operatorname{sym}}_{AB}\log_2\!\left(p\frac{2}{d(d+1)} \Pi^{\operatorname{sym}}_{AB}\right)\right]
		\\
  -\operatorname{Tr}\!\left[\frac{2(1-p)}{d(d-1)}\Pi^{\operatorname{asym}}_{AB}\log_2\!\left(\frac{2(1-p)}{d(d-1)} \Pi^{\operatorname{asym}}_{AB}\right)\right].
	\end{multline}
	The eigenvalues of a projection operator are either equal to one or zero. Hence, for an arbitrary scalar $\alpha$ and an arbitrary projection operator $\Pi^x$,
	\begin{equation}\label{eq:proj_op_log_id}
		\log_2\!\left(\alpha \Pi^x\right) = \Pi^x\log_2\alpha.
	\end{equation}	 
	Therefore,
	\begin{multline}\label{eq:VN_ent_werner_simp}
		H\!\left(W^{p,d}_{AB}\right) = -p\frac{2}{d(d+1)}\log_2\!\left(p\frac{2}{d(d+1)}\right)\operatorname{Tr}\!\left[\Pi^{\operatorname{sym}}_{AB}\right] \\
  - (1-p)\frac{2}{d(d-1)}\log_2\!\left((1-p)\frac{2}{d(d-1)}\right)\operatorname{Tr}\!\left[\Pi^{\operatorname{asym}}_{AB}\right].
	\end{multline}
	From \eqref{eq:part_tr_sym_op} and \eqref{eq:part_tr_asym_op}, it is straightforward to see the following equalities:	
	\begin{align}
		\operatorname{Tr}\!\left[\Pi^{\operatorname{sym}}_{AB}\right] &= \frac{d(d+1)}{2},\\
		\operatorname{Tr}\!\left[\Pi^{\operatorname{asym}}_{AB}\right] &= \frac{d(d-1)}{2}.
	\end{align}
	Substituting these values in \eqref{eq:VN_ent_werner}, we arrive at the following expression for the von Neumann entropy of the Werner state:
	\begin{align}
		H\!\left(W^{p,d}_{AB}\right) &= -p\log_2\!\left(p\frac{2}{d(d+1)}\right) - (1-p)\log_2\!\left((1-p)\frac{2}{d(d-1)}\right)\\
        &= h_2(p) + p\log_2(d(d+1)) - p +(1-p)\log_2(d(d-1)) -(1-p)\\
        &= \log_2 d + h_2(p) - 1 +p\log_2 (d+1)+ (1-p)\log_2 (d-1),\label{eq:VN_ent_werner}
	\end{align}
    where $h_2(p)\coloneqq -p\log_2p - (1-p)\log_2(1-p)$.
	Combining \eqref{eq:coherent_info_defn}, \eqref{eq:VN_ent_werner_partial}, and \eqref{eq:VN_ent_werner}, we conclude \eqref{eq:coherent_info_werner}.
    \end{proof}

\section{Coherent information of isotropic states}\label{app:coherent_info_iso_proof}

\begin{proposition}
    Consider a $d$-dimensional isotropic state $\zeta^{F,d}_{AB}$ which is defined as follows:
    \begin{equation}
        \zeta^{F,d}_{AB} = F~\Phi^d_{AB} + (1-F)\frac{I_{AB}-\Phi^d_{AB}}{d^2-1},
    \end{equation}
    where $F\in [0,1]$ and $d = d_A = d_B \in \{2,3,4, \ldots\}$. The coherent information of the state $\zeta^{F,d}_{AB}$ is equal to the following quantity:
    \begin{equation}\label{eq:coherent_info_iso}
        I\!\left(A\rangle B\right)_{\zeta^{F,d}} = \log_2 d - h_2(F) - (1-F)\log_2(d^2-1),
    \end{equation}
    where $h_2(x) \coloneqq -x\log_2 x - (1-x)\log_2(1-x)$.
\end{proposition}

\begin{proof}
    Recall the definition of coherent information of a bipartite state $\rho_{AB}$ from \eqref{eq:coherent_info_defn}:
    \begin{equation}\label{eq:coherent_info_defn_iso}
		  I\!\left(A\rangle B\right)_{\rho} \coloneqq H\!\left(\operatorname{Tr}_A\!\left[\rho_{AB}\right]\right) - H\!\left(\rho_{AB}\right).
	   \end{equation}
    The marginal of the isotropic state on either systems, $A$ or $B$, is equal to the maximally mixed state for all $F\in [0,1]$ and $d\ge 2$. Therefore,
    \begin{equation}\label{eq:ent_iso_marginal}
        H\!\left(\operatorname{Tr}_B\!\left[\zeta^{F,d}_{AB}\right]\right) = H\!\left(\frac{I_{A}}{d}\right) = \log_2 d.
    \end{equation}

    Now let us evaluate the von Neumann entropy of the isotropic state. Note that $\Phi^d_{AB}$ and $I_{AB}-\Phi^d_{AB}$ are orthogonal projectors. This implies the following equalities:
    \begin{align}
        H\!\left(\zeta^{F,d}_{AB}\right)
        &= -\operatorname{Tr}\!\left[\zeta^{F,d}_{AB}\log_2\zeta^{F,d}_{AB}\right]\\
        &= -\operatorname{Tr}\!\left[F\Phi^d_{AB}\log_2\!\left(F\Phi^d_{AB}\right)\right]- \operatorname{Tr}\!\left[\frac{1-F}{d^2-1}\left(I_{AB}-\Phi^d_{AB}\right)\log_2\!\left(\frac{1-F}{d^2-1}\left(I_{AB}-\Phi^d_{AB}\right)\right)\right]\\
        &= -F\log_2F\operatorname{Tr}\!\left[\Phi^d_{AB}\right] - \frac{1-F}{d^2-1}\log_2\!\left(\frac{1-F}{d^2-1}\right)\operatorname{Tr}\!\left[I_{AB}-\Phi^d_{AB}\right]\\
        &= -F\log_2 F - (1-F)\log_2\!\left(\frac{1-F}{d^2-1}\right)\\
        &= -F\log_2 F - (1-F)\log_2(1-F) + (1-F)\log_2\!\left(d^2-1\right)\\
        & = h_2(F) + (1-F)\log_2\!\left(d^2-1\right), \label{eq:ent_iso_st}
    \end{align}
    where we have used the property of projection operators mentioned in \eqref{eq:proj_op_log_id} to arrive at the third equality. Substituting \eqref{eq:ent_iso_marginal} and \eqref{eq:ent_iso_st} in the definition of coherent information, we arrive at \eqref{eq:coherent_info_iso}. This concludes the proof.
\end{proof}
	
\end{document}